\documentclass[12pt,leqno]{amsart}
\usepackage{amsmath,amsthm}
\usepackage{amsfonts}
 \newtheorem{res}{Result}[section]
 \newtheorem{thm}[res]{Theorem}
 \newtheorem{ex}[res]{Example}

 \newtheorem{rem}[res]{Remark}

\newtheorem{prop}[res]{Proposition}
 \newtheorem{lem}[res]{Lemma}
 \newtheorem{cor}[res]{Corollary}
 \newtheorem{defi}[res]{Definition}

\numberwithin{equation}{section}

\def\bF{\mathbf F}
\def\k{\mathbf {lin}}
\def\op{\oplus}
\def\tgamma{\beta}

\def\Y{\mathcal Y}
\def\v{\mathbf V}
\def\dd{\mathbf d}

\def\ttheta{\tilde \theta}

\def\0{\mathbf 0}
\def\tp{\tilde \P}

\def\pice{\Pi^{c,e}}
\def\w{\mathbf w}
\def\d{\Delta}
\def\dd{\mathbf d}

\def\p{\mathbf P}

\def\Ml{{\mathcal M}_{loc}}
\def\V{\mathbf V}
\def\Y{\mathcal Y}

\def\q{\mathcal Q}

\def\B{\mathcal B}

\def\C{\mathcal C}
\def\D{\mathcal D}
\def\m{\mathbf M}
\def\cc{\mathbf C}

\def\T{[0,T]}

\def\b{b{\mathcal F}^+}
\def\B{\mathcal B}

\def\a{\mathbf a}
\def\b{\mathbf b}

\def\eqf{\buildrel F\over =}

\def\L{\mathcal L}

\def\cp{\mathcal P}

\def\Corr{\rho}

\def\L{\mathbb L}

\def\cP{\mathcal P}
\def\A{\mathcal A}

\def\tB{\tilde \B}

\def\al{\alpha}

\def\1{\mathbf 1}
\def\F{\mathcal F}

\def\R{\mathbb R}
\def\P{\mathbb P}

\def\bbQ{\mathbb Q}

\def\bQ{\bbQ}
\def\bq{\bQ}

\def\Lam{\Lambda}

\def\lam{\lambda}

\def\Linf{\L^{\infty}}

\def\M{\mathcal M}
\def\Ml{\M_{loc}}
\def\te{\tilde \e}
\def\A{\mathcal A}

\def\al{\alpha}

\def\1{\mathbf 1}
\def\G{\mathcal G}
\def\F{\mathcal F}

\def\1{\mathbf 1}

\def\cE{\mathcal E}

\def\e{\mathbb E}
\def\eq{{\mathbb E}^{\bQ}}

\def\d+{\Delta^+}
\def\ce{\cE}

\begin{document}
\bibliographystyle{plain}
\title[]{On the degree of incompleteness of an incomplete financial market.}
\author{Abdelkarem Berkaoui}
\address{Department of mathematics and statistics, Science college, Al-Imam Mohammed Ibn Saud Islamic University (IMSIU), P.O. Box 90950,
Riyadh 11623, Saudi Arabia.}
\email{aamberkaoui@imamu.edu.sa}

\begin{abstract}
In order to find a way of measuring the degree of incompleteness of an incomplete financial market, the rank of the vector price process of the traded assets and the dimension of the associated acceptance set are introduced. We show that they are equal and state a variety of consequences.
\end{abstract}


\maketitle


\section{Introduction, notation and review.}

\subsection{Introduction.}

The notion of an equivalent martingale measure is a corner stone within the theory of financial mathematics. Effectively the two fundamental principles of asset pricing are built upon this notion. More precisely the first principle says that a market, which is composed of a num\'eraire and a finite set of traded assets, satisfies the no-arbitrage property if and only if the price process of these assets admits at least one equivalent martingale measure, while the second principle precise further that any contingent claim in this market is attainable, which means that it can be hedged via a self-financing strategy, if and only if this equivalent martingale measure is unique. The market is then called incomplete in the first case, and complete in the second one. We refer to \cite{DS} and to references therein for more details on the subject.

Without any doubt, the complete market is the exemplary market where the risk is completely avoided. So starting from an incomplete market, can we measure the degree of incompleteness in this market. Or in other words, can we measure the distance from the nearest complete market if it exists. One way to answer this question, at least theoretically is to look at the minimum number of traded assets we have to add to this market in order to increase our trading options, lower the existing risk and then reach a complete market. Another way of looking at it is to precise the minimum number of traded assets that generate the involved market.

To illustrate this, let consider a one risky asset, whose price $X$ is expressed in the Brownian setting as solution of the stochastic differential equation $dX=\sigma X\,dW$ and $\sigma$ is a bounded adapted process. This market is complete iff the set $\{\sigma=0\}$ has null probability. So it can be looked at this market as follows: The process $W$ is a model for the risk existing in the market and the process $X$ has to generate the whole Brownian filtration $\F^W$ in order to reach the completeness of the market.

It may seem natural to start from the fact a market can be modeled by a vector price process $X$ satisfying mainly the no-arbitrage assumption and other technical assumptions. We associate to it the cone of attainable claims $B(X)=\{\al\bullet X_T:\;\al\;\mbox{is an admissible strategy}\}$, where $T$ is the time horizon or maturity time and define $\A=\A(X)=(B(X)-\L^0_+)\cap\Linf$. We define the rank of $X$ as the minimum number $r$ such that there exists a vector price process $Y=(Y^1,\ldots,Y^r)$ satisfying $\A(X)=\A(Y)$. In a completely different approach and since the set $\A(X)$, defined earlier satisfies all the axioms of an acceptance set and following the terminology of the theory of coherent risk measures, we denote $\b$ to be the set of all acceptance sets $\B$ such that the associated set $\q^{\B}$ of test probabilities is the set of martingale measures for a family of adapted processes. So in order to define the dimension of an element $\A\in\b$, we should define first the components that constitute $\A$ and define the dimension of $\A$ as the minimum number of these components needed to generate the set $\A$. In \cite{berk}, the notion of a section of $\A$ was introduced and it was shown that only a building block of $\A$ is a section of $\A$. So on the basis of this concept we define a component as a set $\A\in\b$ such that any subset $\B\in\b$ in $\A$ is a section of $\A$ and define the dimension of a general set $\A\in\b$ as the minimum number of components $\A^1,\ldots,\A^d$ such that $\A=\A^1+\ldots+\A^d$.

We shall show later that the rank $r$ of $X$ and the dimension $d$ of $\A(X)$ are the same. Additionally to that we associate to the vector price process $X$, a predictable partition of the sample time space such that on each unit $G$ of this partition, the process $\1_F\bullet X$ has the same rank for all $F\subseteq G$. We will also introduce the notions of a complement and a strict complement sets as analogues to the concept of orthogonal vector space, and so by fixing two sets $\A,\B\in\b$ with $\B\subseteq\A$, we can decompose the set $\A$ into the sum of $\B$ and a minimal set $\C\in\b$. We also characterize the case where the dimension of $\A$ is the sum of the dimensions of $\B$ and $\C$.

For further analogy with the algebraic dimension of a vector space, we introduce the plug-in vector space $\V(X)$ of $X$ as the closure in $\L^0$ of the set of integrands of X, and show that the associated random algebraic dimension $\dd(X)$ is closely related to the rank and to the earlier partition. This random algebraic dimension is used to define a unit of measure for the completeness of $X$ as a generator of a financial market. Further consequences and applications of this notion are given.


\subsection{Notation.}

Let consider a stochastic basis $(\Omega,\F,(\F_t)_{t\in\T},\P)$, satisfying the usual conditions. An acceptance set $\A$ is defined as a weak star closed convex cone in $\L^{\infty}$, satisfying the two conditions: $\Linf_-\subseteq\A$ and $\A\cap\Linf_+=\{0\}$. We associate to it the set $\q=\q^{\A}$ of $\P$-absolutely continuous probability measures, taking negative values on $\A$. We denote by $\a$, the set of all acceptance sets and by $\a^{st}$ the set of elements $\A\in\a$ for which the associated set $\q^{\A}$ is m-stable.

Along this paper we will adopt the following notation:

\begin{itemize}

\item $\p$ denotes the set of all $\P$-absolutely continuous probability measures and $\p^e$ is the subset in $\p$ of the equivalent ones.
\item $\Pi$ denotes the set of subsets of $\p$ and $\pice$ is the set of closed convex sets $\q\in\Pi$ such that $\q\cap\p^e\neq \emptyset$.
\item $\cp$ denotes the predictable $\sigma$-algebra generated by the class of left continuous processes with limits at right.
\item For a set $\q$ of probability measures, we denote $spm(\q)$, to be the set of $\q$-super martingales and $m(\q)$, to be the set of local $\q$-martingales.
\item For a family $\Y$ of adapted processes, we denote by $\M_{sp}(\Y)$ (resp. $\Ml(\Y)$), the set of all local spermartingale (resp. martingale) measures of the family $\Y$.
\item $\L^p(\G;\R^d)$ denotes the Lebesgue space of $\R^d$-valued, $p$-integrable and $\G$-measurable random variables with $\L^p(\G)=\L^p(\G;\R)$ and $\L^p=\L^p(\F)$, for $1\leq p\leq\infty$ and a sub-$\sigma$-algebra $\G\subseteq\F$.
\item For $B\subseteq\L^p$, we denote $\overline{B}^p$ the closure in $\L^p$ of $B$, and $\overline{B}^*$ for the weak star closure in $\Linf$.
\item $\al\bullet S$ denotes the stochastic integral process of the process $\al$ w.r.t. the semimartingale $S$.
\item $\m_{n,m}(\cp)$ denotes the set of all $\R^n\otimes\R^m$-matrix valued predictable processes and for $\al\in\m_{n,m}(\cp)$, $span(\al)$ denotes the vector space spanned by the rows of the matrix $\al$. We shall say that two matrices are orthogonal if their spanned vector spaces are. We denote by $rank(\al)$, the random algebraic rank of the matrix $\al$.
\item $\1_F$ and $F^c$ denote respectively the indicator function and the complementary set of a measurable set $F$.

\end{itemize}


\subsection{Review.}
We recall the notion of a section of an acceptance set introduced in \cite{berk} and review some of its properties.
\begin{defi}\label{d1}(Berkaoui \cite{berk})Let $F\in\cp$ and $\A\in\a$ with the associated set $\q$ of test probabilities. We define the $F$-section of $\A$ as follows:
$$
\1_F\circ \A:=\{h\in\L^\infty:\;h\leq \1_F\bullet X_T\;\mbox{for some}\;X\in spm(\q)\}.
$$

It has been proved in \cite{berk} that $\1_F\circ\A\in\a$, we denote by $\1_F\circ\q$, the associated set of test probabilities. We shall say that $\B\in\a$ is a section of $\A$ if it is an $F$-section of $\A$ for some $F\in\cp$.
\end{defi}

\begin{thm}\label{xt1}(Berkaoui \cite{berk})Let $F,G\in\cp$ and $\A,\B\in\a$ such that $\A+\B\in\a$. Then
\begin{enumerate}
\item $\1_F\circ(\1_G\circ\A)=\1_{F\cap G}\circ\A$.
\item $\1_F\circ\A+\1_{G}\circ\A=\1_{F\cup G}\circ\A$.
\item $\1_F\circ(\A+\B)=\1_F\circ\A+\1_F\circ\B$.
\item $\1_F\circ\M_{sp}(\Y)=\M_{sp}(\1_F\bullet\Y)$ where $\1_F\bullet\Y=\{\1_F\bullet Y:\;Y\in \Y\}$ for any family $\Y$ of adapted processes.
\item $\1_F\circ\Ml(\Y)=\Ml(\1_F\bullet\Y)$.

\end{enumerate}

\end{thm}



\section{Main results.}

\subsection{Definitions.}

We define $\b$ to be the set of all elements $\A\in\a$ such that $\q^{\A}$ is the set of martingale measures for a family of adapted processes, and denote by $\b(\A)$, the set of all acceptance subsets of $\A$, which are elements in $\b$.

Now we introduce the notion of dimension of an element $\A\in\b$.

\begin{defi}\label{d2}Let $\A\in\b$.
\begin{enumerate}
\item We say that $\A$ is a trivial set if $\A=\Linf_-$.

\item We say that a non trivial set $\A$ is a component if any element $\B\in\b(\A)$ is a section of $\A$.

\item We say that a non trivial set $\A$ is of class $(\Delta)$ if it admits a special decomposition of order $n$ for some positive integer $n$, which means that there exists $n$ subcomponents $\A^1,\ldots,\A^n$ of $\A$ such that $\A=\A^1+\ldots+\A^n$.

\item Suppose $\A$ is of class $(\Delta)$, we define the dimension of $\A$ as the minimum positive integer $n$ such that $\A$ admits a special decomposition of order $n$. For a trivial set $\A=\Linf_-$, we set $dim(\A)=0$.

\item We say that $\A$ is of class $(\Delta_n)$ if the dimension of $\A$ is equal to $n$.
\end{enumerate}

\end{defi}

\begin{rem}\label{r1}From Definition \ref{d2}, a component $\A$ is of class $(\Delta_1)$.
\end{rem}

We denote by $\b_n$, the set of elements in $\b$ which are of class $(\Delta_n)$ and by $\b_n(\A)$, the set of elements in $\b_n$, which are subsets in $\A$.

In this section we characterize elements in $\b_n$ for $n\geq 1$. We define $\cc_n$, to be the set of $\R^n$-valued adapted processes admitting at least one equivalent local martingale measure, and for $X\in\cc_n$, we denote $\q^X:=\Ml(X)$ to be the set of all local martingale measures of $X$, $\A(X)$ the associated acceptance set and $m(\q^X)$ to be the set of all local $\q^X$-martingales. We shall say that $X\in\cc_n$ is a generator of $\A$ if $\A=\A(X)$.

\subsection{Component characterization.}

We start by stating the following results.
\begin{prop}\label{p2}
Let $\A\in\b$ and let $\q=\q^{\A}$ the associated set of test probabilities. Let define the set
$$
\Phi=\{F=F^X\in\cP:\;\1_F\,\circ\A=\A(X);\;\mbox{for some}\;X\in m(\q)\}.
$$
Then $\Phi$ admits a maximum element w.r.t. inclusion order.
\end{prop}

\begin{proof}
We will show the three assertions of Lemma 4.15 in \cite{berk}. For (i) let $F,G\in\Phi$ with $F\cap G=\emptyset$, then $F=F^X$ and $G=F^Y$ for some $X,Y\in m(\q)$. Therefore we have $\1_{(F\cup G)}\circ\A=\1_{F}\circ\A+\1_{G}\circ\A=\1_{F}\circ\A(X)+\1_{G}\circ\A(Y)=\A(U)$ with $U=\1_{F}\bullet X+\1_{G}\bullet Y\in m(\q)$, so $F\cup G\in\Phi$.

For (ii) let an increasing sequence $(F^n)\subseteq\Phi$, so $F^n=F^{X^n}$ for some $X^n\in m(\q)$ with $F=\cup_{n\geq 1}F^n$. Then $\1_{F}\circ\A=\overline{\cup_{n\geq 1}\1_{F^n}\circ\A}^*$ thanks to Lemma 4.16 in \cite{berk} and the closure is taken in the weak star sense. So $\1_{F}\circ\A=\overline{\cup_{n\geq 1}\1_{G^n}\circ\A(X^n)}^*$ where $G^1=F^1$ and $G^{n+1}=F^{n+1}\backslash F^n$ for all $n\geq 1$. Therefore $\1_{F}\circ\A=\A(X)$ where $X=\sum_{n\geq 1}\1_{G^n}\bullet X^n \in m(\q)$. We conclude that $F\in\Phi$.

For (iii) let $F\in\Phi$ and $G\subseteq F$, then $F=F^X$ for $X\in m(\q)$ and therefore $\1_{G}\circ\A=\1_{G}\1_{F}\circ\A=\1_{G}\circ\A(X)=\A(Y)$ with $Y=\1_{G}\bullet X\in m(\q)$. So $G\in\Phi$.

\end{proof}

\begin{prop}\label{xp1}
Let $X,Y\in\cc_1$ such that $Y=\al\bullet X$ for some scalar bounded predictable process $\al$. Then $\A(Y)=\1_F\circ\A(X)$ where $F=supp(\al):=\{|\al|>0\}$.
\end{prop}

\begin{proof}
We will show that $m(\q^{Y})=\1_F\bullet m(\q^X)$. For the direct inclusion we have $m(\q^{Y})=\1_F\bullet m(\q^{Y})\subseteq\1_F\bullet m(\q^X)$. Conversely let $U\in m(\q^X)$ and $V=\1_F\bullet U$. Then thanks to Jacka's Theorem in \cite{jacka}, we get $U=\gamma\bullet X$ for some predictable process $\gamma$ and $V=\1_F(1/\al)\al\bullet U=\1_F(1/\al)\al\gamma\bullet X=\1_F(1/\al)\gamma\bullet Y\in m(\q^{Y})$.
\end{proof}

Now we characterize elements of $\b$ which are of class $(\Delta_1)$.
\begin{prop}\label{p1}
Suppose $\A\in\b$. Then $\A$ is a component if and only if $\A=\A(X)$ for some $X\in\cc_1$.

\end{prop}

\begin{proof}For the direct implication let $U\in m(\q)$ and since $\A(U)\in\b(\A)$, then there exists some $F=F^U\in\cp$ such that $\A(U)=\1_F\circ\A$. We define the set $\Phi=\{F^U:\;U\in m(\q)\}$. Thanks to Proposition \ref{p2}, the set $\Phi$ admits a maximum element $F=F^X$ w.r.t inclusion order. We claim that $\A=\A(X)$. First $\A(X)\subseteq\A$, let us show that $\A\subseteq\A(X)$ by showing that $m(\q)\subseteq m(\q^X)$. Let $Y\in m(\q)$, then $\A(Y)=\1_{F^Y}\circ\A\subseteq\1_F\circ\A=\A(X)$ and therefore $Y\in m(\q^X)$.

Conversely let $\B\in\b(\A)$, then we get $\q^{\B}=\Ml(\al\bullet X)$ for a predictable process $\al$ thanks to Theorem 3.3 in \cite{berk2}. We apply Proposition \ref{xp1} and conclude the result.

\end{proof}

\begin{prop}\label{p4}Let an integer $n\geq 1$. Then $\A\in\b$ has a special decomposition of order $n$ if and only if $\A=\A(X)$ for some $X\in\cc_n$.

\end{prop}

\begin{proof}
For the direct implication, let us suppose that $\A=\A^1+\ldots+\A^n$ for a family of subcomponents $\A^1,\ldots,\A^n$ in $\A$. Thanks to Proposition \ref{p1}, there exists scalar processes $X^1,\ldots,X^n\in\cc_1$ such that $\A^i=\A(X^i)$ for $i=1\ldots n$ and therefore $\A=\A(X)$ with $X=(X^1,\ldots,X^n)\in\cc_n$.

Conversely suppose that $\A=\A(X)$ for $X\in\cc_n$, then thanks to Lemma 4.3 in \cite{berk}, we get $\A=\A(X^1)+\ldots+\A(X^n)$, which means that $\A$ has a special decomposition of order $n$.

\end{proof}

\subsection{Class $(\Delta_n)$ characterization.}

In order to characterize the class $(\Delta_n)$ for any positive integer $n$, we introduce some definitions. We recall that $\cc_n$ is the set of $\R^n$-valued adapted processes admitting at least one equivalent local martingale measure and define $\cc_{n,r}$ for $r\leq n$, to be the subset in $\cc_n$ of $\R^n$-valued adapted processes $X$ such that $X=\al\bullet Y$ for some $Y\in\cc_r$ and $\al\in\m_{n,r}(\cp)$. We set $\cc:=\cup_{n\geq 1}\cc_n$.

\begin{defi}\label{d3} Let $X\in\cc_n$. We define the rank of $X$ by
$$r(X)=\min\{r=1\ldots n:\;X\in\cc_{n,r}\}.$$
\end{defi}

Now we state the following

\begin{thm}\label{tt1}Let $X\in\cc_n$. Then $r(X)=dim(\A(X))$.

\end{thm}

\begin{proof} We show first that for all $r\leq n$, we have $dim(\A(X))\leq r$ if and only if $r(X)\leq r$. We suppose $dim(\A(X))\leq r$, then $\A(X)=\A(Y)$ for some $Y\in\cc_r$ and since $X\in m(\q^X)=m(\q^Y)$, then $X=\al\bullet Y$ for some $\al\in\m_{n,r}(\cp)$, so $r(X)\leq r$.
Suppose now $r(X)\leq r$, then $X=\al\bullet Y$ for some $\al\in\m_{n,r}(\cp)$ and $Y\in\cc_r$. The set $K$ defined by $K=\overline{\{\beta\al:\;\beta\in\L^0(\cp;\R^n\}}^0$, is a closed vector space in $\L^0(\cp,\R^r)$ and closed under multiplication by bounded positive $\cp$-measurable random variables, so there exists a closed vector space valued measurable mapping $W$ such that $K=\{\gamma\in\L^0(\cp;\R^r):\;\gamma\in W\;\mbox{a.s.}\}$ and there exists then a generating family $f=(f^1,\ldots,f^r)$ of $W$. We will show that $\A(X)=\A(U)$ where $U=g\bullet Y$, $g=f/(1+|f|)$ and deduce that $dim(\A(X))=dim(\A(U))\leq r$. In fact for any $R\in m(\q^X)$, we have $R=\delta\bullet X=\delta\al\bullet Y=\delta' g\bullet Y=\delta'\bullet U$. Inversely for any $R\in m(\q^U)$, we have $R=\theta\bullet U=\theta g\bullet Y=\theta'\al\bullet Y=\theta'\bullet X$.

For the equality $r(X)=dim(\A(X))$, we take first $r=r(X)$ and deduce that $dim(\A(X))\leq r(X)$, and second we take $r=dim(\A(X))$ and deduce that $r(X)\leq dim(\A(X))$.
\end{proof}

Some consequences are given below.

\begin{cor}\label{cc2}Let $X\in\cc_n$ and $r\leq n$. Then the following assertions are equivalent:
\begin{enumerate}
\item $\A(X)$ is of class $(\Delta_r)$.
\item $r(X)=r$.
\item $\A(X)=\A(Y)$ for some $Y\in\cc_r$ with $r(Y)=r$.

\end{enumerate}
\end{cor}

\begin{cor}\label{c2}Let $\A\in\b$ and a positive integer $n$. Then $\A$ is of class $(\Delta_n)$ if and only if $\A=\A(X)$ for some $X\in\cc_n$ and $r(X)=n$.
\end{cor}

Finally we investigate the complete financial market case.

\begin{thm}\label{t5}Let consider an integer $n$. Then the following assertions are equivalent:
\begin{enumerate}
\item $\{\P\}=\Ml(X)$ for some $X\in\cc_n$.
\item $dim(\B)\leq n$ for all $\B\in\b$.
\end{enumerate}
Moreover $r(X)=n$ if and only if there exists some $\B^0\in\b_n$.

\end{thm}

\begin{proof}
$(1)\Rightarrow(2)$ Let $\B\in\b$, so thanks to Theorem 3.12 in \cite{berk2}, there exists some $Y\in\cc_n$ such that $\B=\A(Y)$ and then $dim(\B)=r(Y)\leq n$.

$(2)\Rightarrow(1)$ Thanks to Corollary 2.4 in \cite{berk2}, we deduce that $\A:=\A^{\{\P\}}\in\b$. Then there exists some $X\in\cc_n$ such that $\A=\A(X)$ and therefore $\{\P\}=\Ml(X)$.

Now suppose $r(X)=n$, then $\B^0:=\A^{\{\P\}}\in\b_n$. Conversely we suppose by absurd that $r(X)=r<n$, so $\{\P\}=\Ml(Y)$ for some $Y\in\cc_r$ and by applying the implication $(1)\Rightarrow(2)$, we deduce that $dim(\B)\leq r$ for any $\B\in\b$. This is a contradiction.

\end{proof}


\section{The predictable ranking map and maximality.}
\subsection{The predictable ranking map.}
In this subsection, we investigate much deeper the notion of rank. Suppose in the Brownian setting with $W=(W^1,W^2)$, we define the process $X=(W^1,\1_F\bullet W^2)$ for some $F\in\cp$, then $r(X)=2$. But $\1_{F^c}\bullet X=(\1_{F^c}\bullet W^1,0)$ and therefore $r(\1_{F^c} \bullet X)=1$. Next we characterize the predictable partition of the time space $\Omega\times\T$, in which the rank of a vector process $X\in\cc_n$ does not change.

\begin{prop}\label{p11}Let $X\in\cc_n$, then
\begin{enumerate}
\item there exists a partition $(\phi_k(X))_{k=0\ldots n}\subseteq\cp$ such that $r(\1_{G}\bullet X)=k$ for all $G\subseteq \phi_k(X)$ and $k=0\ldots n$.
\item there exists some $U\in\cc_n$ such that $\A(X)=\A(U)$ and for all $k=1\ldots n$, we have $\1_{\phi_k(X)}\circ\A(X)=\1_{\phi_k(X)}\circ\A(U^1,\ldots,U^k)$ and $\phi_k(X)=\phi_k(U^1,\ldots,U^k)$.

\end{enumerate}
\end{prop}

\begin{proof} (1) We define the set $\Psi_k:=\{F\in\cp:\;r(\1_F\bullet X)\leq k\}$ for $k<n$. We will show that the set $\Psi_k$ admits a maximum element. To do so we remark first that $F\in\Psi_k$ if and only if $\1_F\bullet X=\1_F\,\al^F\bullet U^F$ for some $\al^F\in\m_{n,k}(\cp)$ and $U^F\in\cc_k$, and second we will show the three assertions of Lemma 4.15 in \cite{berk}. For (i) let $F,G\in\Psi_k$ with $F\cap G=\emptyset$, then $\1_{(F\cup G)}\bullet X=\1_F\bullet X+\1_G\bullet X=\1_F\al^F\bullet U^F+\1_G\al^G\bullet U^G=\1_{F\cup G}\al\bullet U$ with $\al=\1_F\al^F+\1_G\al^G\in\m_{n,k}(\cp)$ and $U=\1_F\bullet U^F+\1_{G}\bullet U^G\in\cc_k$, so $F\cup G\in\Psi_k$. For (ii) let an increasing sequence $F^n\in\Psi_k$ with $F=\cup_{n\geq 1}F^n$. We define the sequence $G^1=F^1$ and for all $n\geq 1$, $G^{n+1}=F^{n+1}\backslash F^n$, then $\1_{F}\bullet X=\1_F\al\bullet U$ with $\al=\sum_{n\geq 1}\1_{G^n}\al^{F^n}\in\m_{n,k}(\cp)$ and $U=\sum_{n\geq 1}\1_{G^n}\bullet U^{F^n}\in\cc_k$. We conclude that $F\in\Psi_k$. For (iii) let $F\in\Psi_k$ and $G\subseteq F$. Then $\1_{G}\bullet X=\1_{G}\1_{F}\bullet X=\1_{G}\1_{F}\al^F\bullet U^F=\1_G\al^F\bullet U^F$, so $G\in\Psi_k$.

Now for $k=0\ldots n-1$ we denote $G^k$ to be the maximum element of $\Psi_k$. We define the family $(\phi_k(X))_{k=0\ldots n}$ by $\phi_n(X)=(G^{n-1})^{c}$, $\phi_k(X)=G^k\backslash G^{k-1}$ for $k=1\ldots n-1$ and $\phi_0(X)=G^0$. So for $G\subseteq \phi_n(X)$ we have $r(\1_{G}\bullet X)=n$ since $G\subseteq (G^{n-1})^c$ and for $G\subseteq \phi_k(X)$ and $k=0\ldots n-1$ we have $r(\1_{G}\bullet X)\leq k$ since $G\subseteq G^k$ and $r(\1_{G}\bullet X)> k-1$ since $G\subseteq (G^{k-1})^{c}$, therefore $r(\1_G\bullet X)=k$.

(2) For all $k=1\ldots n$, there exists some $(V^{k1},\ldots,V^{kk})\in\cc_k$ such that $\1_{\phi_k(X)}\circ\A(X)=\1_{\phi_k(X)}\circ\A(V^{k1},\ldots,V^{kk})$. We define the vector process $U\in\cc_n$ by $U^j=\sum_{k=j}^n\,\1_{\phi_k(X)}\bullet V^{kj}$ for $j=1\ldots n$, and conclude the result.

\end{proof}

\begin{defi}\label{d5}Following the notations of Proposition \ref{p11}, the family $\phi(X)=(\phi_k(X))_{k=0\ldots n}$ is called the predictable ranking map of $X$, and the process $U$ is called the $\phi(X)$-arrangement of the process $X$.
\end{defi}

\begin{rem}\label{r3}It has been proved in \cite{berk} that for any $\A\in\a$, there exists a maximum element $F=F(\A)\in\cp$ such that $\1_F\circ\A=\Linf_-$. In particular for $\A=\A(X)$ for some $X\in\cc_n$ and thanks to Proposition \ref{p11}, we have that $F(\A)=\phi_0(X)$.

\end{rem}

An illustrative example is given below.
\begin{ex}\label{ex1}Let consider the Brownian setting $W=(W^1,W^2)$ and define $X=(\1_F\bullet W^1,W^2)$ for some predictable set $F$. Then $\phi_0(X)=\emptyset$, $\phi_1(X)=F^{c}$, $\phi_2(X)=F$, and the $\phi(X)$-arrangement process $U$ is given by $U^1=\1_{F}\bullet W^1+\1_{F^{c}}\bullet W^2$ and $U^2=\1_{F}\bullet W^2$.

\end{ex}

For an integer $n\geq 1$, we will say that a process $X\in\cc_n$ satisfies the property $(\gamma)$ if for all $\lam\in\m_{1,n}(\cp)$ such that $\lam\bullet X=0$, then $\lam\equiv 0$. This property is similar to the linear independence property of vectors, well known in linear algebra, by taking the coefficients being the integrands and the operation being integration. Next we state two equivalent statements of this property.

\begin{prop}\label{xp5}Let $\A=\A(X)$ for some $X\in\cc_n$. Then the following assertions are equivalent:
\begin{enumerate}
\item[(1)] $X$ satisfies the property $(\gamma)$.
\item[(2)] $\1_F\circ\A\in\b_n$ for all $F\in\cp$.
\item[(3)] $\tp(\phi_n(X))=1$.
\end{enumerate}

\end{prop}

\begin{proof}
$(1)\Rightarrow(3)$ Suppose $\tp(\phi_n(X))<1$, so $\tp(\phi_i(X))>0$ for some $i<n$. Therefore for $F^i=\phi_i(X)$ we get $\1_{F^i}\circ \A=\A(U)$ for some $U\in\cc_i$ and therefore $\1_{F^i}\bullet X=M\bullet U$ for some $M\in\m_{n,i}(\cp)$. We take a vector process $\lam\in\m_{1,n}(\cp)$, non null on $F^i$ such that $\lam M =0$, so $\1_{F^i}\lam\bullet X=0$. This is a contradiction.

$(3)\Rightarrow(2)$ Suppose $\1_F\circ\A\in\b_s$ for some $F\in\cp$ and $s<n$. Then $\tp(\phi_s(X))>0$ which is a contradiction.

$(2)\Rightarrow(1)$ Let $\lam\bullet X=0$ for some $\lam\in\m_{1,n}(\cp)$, and suppose that there exists some $i=1\ldots n$ such that $\tp(|\lam^i|>0)>0$. Let $F=\{|\lam^i|>0\}$, so we remark that $\1_F\bullet X^i=\1_F\beta^i\bullet X^{(i)}$ for some $\beta^i\in\m_{1,n-1}(\cp)$ and $X^{(i)}=(X^j:\;j\neq i)$, therefore $\1_F\bullet X=\1_F\beta\bullet X^{(i)}$ for some $\beta\in\m_{n,n-1}(\cp)$. This means that $r(\1_F\bullet X)\leq n-1$ or equivalently that $dim(\1_F\circ\A)<n$, which is a contradiction.

\end{proof}

\begin{prop}\label{pxx}Let $X\in\cc_n$ satisfies the property $(\gamma)$ and $Y=\theta\bullet X\in\cc_r$ for some $\theta\in\m_{r,n}(\cp)$. Suppose that $\tp(rank(\theta)= s)>0$ for some $s\leq r$, then $r(\1_F\bullet Y)= s$ for all $F\in\cp$ with $F\subseteq G:=\{rank(\theta)= s\}$.
\end{prop}

\begin{proof}We show the result first for the case $s=r$. Let $\lam\in\m_{1,r}(\cp)$ such that $\lam\bullet Y=0$, so $\lam\theta\bullet X=0$ which means that $\lam\theta=0$ since $X$ satisfies the property $(\gamma)$. We deduce that $\lam=0$ on $G$, which means that $\1_F \bullet Y$ satisfies the property $(\gamma)$ for all $F\subseteq G$ and then $r(\1_F \bullet Y)=r$.

For the general case, there exists two predictable matrices $\beta\in\m_{r,s}(\cp)$ and $\delta\in\m_{s,n}(\cp)$ such that $\theta=\beta\delta$ on $G$ and $\tp(rank(\delta)=s;G)>0$. Then $Y=\beta\bullet U$ with $U=\delta\bullet X\in\cc_s$ and $r(\1_F\bullet U)=s$ for all $F\subseteq G$, which means that $r(\1_F\bullet Y)=s$.

\end{proof}

\begin{prop}\label{p6}Let $X\in\cc_n$ satisfies the property $(\gamma)$ and $Y=\theta\bullet X\in\cc_r$ for some $\theta\in\m_{r,n}(\cp)$. Then $\phi_k(Y)=\{rank(\theta)=k\}$ for $k=0\ldots r$.

\end{prop}

\begin{proof}We shall show first that $F_s:=\cup_{k=s}^r\,\phi_k(Y)=\{rank(\theta)\geq s\}=:G_s$ for all $s=0\ldots r$. We suppose by absurd that $\tp(F_s\cap (G_s)^{c})>0$ for some $s$, then by Proposition \ref{pxx} we get $r(\1_F\bullet Y)\leq s-1$ for some $F\subseteq F_s\cap (G_s)^{c}$, which contradicts the fact that $r(\1_G\bullet Y)\geq s$ for all $G\subseteq F_s$. We suppose again by absurd that $\tp(G_s\cap (F_s)^{c})>0$ for some $s$, then by Proposition \ref{pxx} we get $r(\1_F\bullet Y)\geq s$ for some $F\subseteq G_s\cap (F_s)^{c}$, which contradicts the fact that $r(\1_G\bullet Y)\leq s-1$ for all $G\subseteq (F_s)^{c}$. We deduce that $F_s=G_s$.

We conclude that $\phi_r(Y):=F_r=G_r:=\{rank(\theta)= r\}$ and that for all $s=r-1\ldots 0$ we have $\phi_s(Y)=F_{s}\cap(F_{s+1})^{c}=G_{s}\cap(G_{s+1})^{c}=\{rank(\theta)=s\}$.
\end{proof}

Two illustrative examples are given below.

\begin{ex}\label{ex2}Let consider the Brownian setting $W=(W^1,W^2)$ and define $X=(\delta\bullet W^1+W^2,W^1+W^2)$ with $\delta\in\Linf(\cp)$. Then $X=
\left(\begin{array}{cc}
\delta & 1\\
1 & 1
\end{array}
\right)
\bullet\left(\begin{array}{c}
W^1\\
W^2
\end{array}
\right)$ and therefore $\phi_0(X)=\emptyset$, $\phi_1(X)=\{\delta=1\}$, $\phi_2(X)=\{\delta\neq 1\}$, and the $\phi(X)$-arrangement process $U$ is given by $U^1=W^1+\1_{(\delta=1)}\bullet W^2$ and $U^2=\1_{(\delta\neq 1)}\bullet W^2$.

\end{ex}

\begin{ex}\label{ex3}Let consider the linear Brownian setting $W$ and define $X$ to be the solution of the stochastic differential equation $dX=b(X)\,dt+\sigma(X)\,dW$ where $b$ and $\sigma$ are two real functions satisfying the usual assumptions for the existence and uniqueness of the solution $X$. Then $X\in\cc_1$ iff $\1_{(\sigma(X)=0)}b(X)=0$. Under this assumption we have $\phi_0(X)=\{\sigma(X)=0\}$ and $\phi_1(X)=\{\sigma(X)\neq 0\}$.

\end{ex}

We say that a process $X\in\cc_n$ satisfies the property $(\tgamma)$ if there exists some $R\in\cc_m$ for $m\geq n$, satisfying the property $(\gamma)$ with $\A(X)\subseteq\A(R)$.

\begin{prop}\label{pt1}Let suppose $X\in\cc_n$ satisfies the property $(\tgamma)$. Then there exists some $Y\in\cc_n$ satisfying the property $(\gamma)$ such that $\A(X)\subseteq\A(Y)$.
\end{prop}

\begin{proof}Let the process $R\in \cc_m$ satisfies the property $(\gamma)$ such that $\A(X)\subseteq\A(R)$, so there exists some $\al\in\m_{n,m}(\cp)$ such that $X=\al\bullet R$. There exists also some $M\in\m_{n,n}(\cp)$ and $\theta\in\m_{n,m}(\cp)$ with $\tp(rank(\theta)=n)=1$ such that $\al=M\theta$. We define $Y=\theta\bullet R$, then $\tp(\phi_n(Y))=\tp(rank(\theta)=n)=1$ thanks to Proposition \ref{p6}, so $Y$ satisfies the property $(\gamma)$ and $X=M\bullet Y$, therefore $\A(X)\subseteq\A(Y)$.

\end{proof}


\subsection{Maximality property.}

The notion of maximality in this subsection and along the paper is w.r.t the order by inclusion. We investigate the relationship between the property $(\gamma)$ and the maximality in $\b_n$.

\begin{prop}\label{p22}Let $X\in\cc_n$ and consider the following assertions:

\begin{enumerate}
\item $X$ satisfies the property $(\gamma)$.
\item $\A:=\A(X)$ is maximal in $\b_n$.
\end{enumerate}
Then $(1)\Rightarrow(2)$. If moreover $X$ satisfies the property $(\tgamma)$, therefore $(2)\Rightarrow(1)$.
\end{prop}

\begin{proof}
$(1)\Rightarrow(2)$ Let us suppose $\A\subseteq\C$ for some $\C=\A(U)\in\b_n$, therefore $X= M \bullet U$ for some $ M \in\m_{n,n}(\cp)$. We claim that $\tp(rank( M )=n)=1$. We suppose we have the opposite, then for $F=\{rank( M )<n\}$ we have $r(\1_F\bullet X)<n$ which is a contradiction. We conclude that $U= M ^{-1}\bullet X$ and so $\A=\C$.

Suppose now that $X$ satisfies the property $(\tgamma)$, then thanks to Proposition \ref{pt1}, there exists some $Y\in\cc_n$ satisfying the property $(\gamma)$ with $\A(X)\subseteq\A(Y)$ and therefore $\A(X)=\A(Y)$. We conclude that $X$ satisfies the property $(\gamma)$.

\end{proof}

\begin{cor}\label{c5}Let $X\in\cc_n$ satisfies the property $(\gamma)$ and let $Y=\theta\bullet X\in\cc_r$ for some $\theta\in\m_{r,n}(\cp)$. Then $Y$ satisfies the property $(\gamma)$ iff $\A(Y)$ is maximal in $\b_r$.
\end{cor}

\begin{proof}
It's an immediate consequence of Proposition \ref{p22} since $Y$ satisfies the property $(\tgamma)$.

\end{proof}

\begin{cor}\label{c4}Let $X\in\cc_n$ satisfy the assumption $(\gamma)$. Then for any $Y\in\cc_r$ such that $(X,Y)\in\cc_{n+r}$ and $r(X,Y)=n$, we have $\A(Y)\subseteq\A(X)$.

\end{cor}
\begin{proof}First $\A(X)\subseteq\A(X,Y)$, $\A(X,Y)\in\b_n$ and $\A(X)$ is maximal in $\b_n$ thanks to Proposition \ref{p22}. Then $\A(X)=\A(X,Y)$ and therefore $\A(Y)\subseteq\A(X)$.

\end{proof}

\begin{cor}\label{c1}Suppose $X$ satisfies the property $(\gamma)$. Then for any $(i_1,\ldots ,i_k)\in\{1,\ldots,n\}^k$ with $i_1<\ldots<i_k$, the set $\A(X^{i_1},\ldots,X^{i_k})$ is maximal in $\b_k$.
\end{cor}

\begin{proof}Since for any permutation $\sigma$ on the set $\{1,\ldots,n\}$ we have that $\A(X)=\A(X^{\sigma_1},\ldots,X^{\sigma_n})$, then we can suppose w.l.g that $i_1=1,\ldots,i_k=k$. Thanks to Proposition \ref{p22}, it suffices to show that $Y:=(X^1,\ldots,X^k)$ satisfies the property $(\gamma)$. Let us suppose $\lam\bullet Y=0$ for some $\lam\in\m_{1,k}(\cp)$, then $(\lam,0)\bullet X=0$, which means that $\lam\equiv 0$ since $X$ satisfies the property $(\gamma)$.

\end{proof}

\begin{prop}\label{pp1}Let $\B$ and $\C$ be two maximal elements in $\b_n$. Then $\1_F\circ\B+\1_{F^{c}}\circ\C$ is maximal in $\b_n$ for any $F\in\cp$.
\end{prop}

\begin{proof}We show first that $\1_F\circ\B+\1_{F^{c}}\circ\C\in \b_n$. Since $\B,\C\in\b_n$, then we have $\B=\A(X)$ and $\C=\A(Y)$ with $X,Y\in\cc_n$. Let $\bq^Z\in\Ml(X)$ and $\bq^U\in\Ml(Y)$ with respective positive martingale densities $Z$ and $U$. We define the process $K$, solution of the linear stochastic differential equation $dK/K=\1_F\,dZ/Z+\1_{F^{c}}\,dU/U$ and $K_0=1$. We shall prove that $\bq^K\in\Ml(V)$ with $V:=\1_F\bullet X+\1_{F^{c}}\bullet Y$ by showing that the process $KV$ is a local martingale. We apply It\^o's formula and obtain that:

$$
d(KV)=K\,dV+V\,dK+d[K,V]
$$
$$
=\1_FK\,dX+\1_{F^{c}}K\,dY+V\,dK+\1_FK/Z\,d[Z,X]+\1_{F^{c}}K/U\,d[U,Y]
$$
$$
=\1_FK/Z\left(Z\,dX+d[Z,X]\right)+\1_{F^{c}}K/U\left(U\,dY+d[U,Y]\right)+V\,dK
$$
$$
=\1_FK/Z\left(d(ZX)-X\,dZ\right)+\1_{F^{c}}K/U\left(d(UY)-Y\,dU\right)+V\,dK.
$$

and since the processes $ZX,Z,UY,U$ and $K$ are local martingales, then $KV$ is also a local martingale.

Now we suppose that $\1_F\circ\B+\1_{F^{c}}\circ\C\subseteq \D$ with $\D\in \b_n$. Then $\1_F\circ\B\subseteq\1_{F}\circ \D$ and therefore $\B\subseteq\1_F\circ \D+\1_{F^{c}}\circ\B\in \b_n$. Since $\B$ is maximal in $\b_n$, we conclude that $\B=\1_F\circ \D+\1_{F^{c}}\circ\B$ and so $\1_F\circ\B=\1_F\circ \D$. We do the same for $\C$ and deduce that $\1_{F^{c}}\circ\C=\1_{F^{c}}\circ \D$. We conclude that $\1_F\circ\B+\1_{F^{c}}\circ\C=\D$.
\end{proof}

\begin{rem}\label{r5}Proposition \ref{pp1} is still true if we replace $\b_n$ by $\b_n(\A)$ with $\A\in\b_m$ for some $m\geq n$.
\end{rem}

\begin{lem}\label{ll1}Let $\A\in\b_n$ and $\B\in\b_r(\A)$ with $r\leq n$. Let us suppose $\B$ is maximal in $\b_r(\A)$, then for any $F\in\cp$, the set $\1_F\circ\B$ is maximal in $\b_j(\1_F\circ\A)$ for some $j=0\ldots r$.
\end{lem}

\begin{proof}First $\1_{F}\circ\B\in\b_j(\1_{F}\circ\A)$ for $j=dim(\1_{F}\circ\B)\in\{0\ldots r\}$. Now let $\D\in\b_j(\1_{F}\circ\A)$ such that $\1_{F}\circ\B\subseteq \D$. Then $\D=\1_F\circ \D$ and $\B\subseteq \D+\1_{F^{c}}\circ\B\in\b_r(\A)$, therefore $\B=\D+\1_{F^{c}}\circ\B$, which means that $\1_{F}\circ\B=\1_{F}\circ \D=\D$.

\end{proof}

\begin{prop}\label{xp12}Let $\A=\A(X)\in\b_n$ and $\B=\A(Y)\in\b_r(\A)$ with $r\leq n$. Then there exists some $\theta\in\m_{r,n}(\cp)$ with $\tp(rank(\theta)=r)=1$ such that $\tB:=\1_G\circ\A+\1_{G^{c}}\circ\A(\theta\bullet X)$ is maximal in $\b_r(\A)$ and contains $\B$ with $G=\cup_{i=0}^r\,\phi_i(X)$.
\end{prop}

\begin{proof}First $Y=\al\bullet X$ for some $\al\in\m_{r,n}(\cp)$ and there exists some $M\in\m_{r,r}(\cp)$ and $\theta\in\m_{r,n}(\cp)$ with $\tp(rank(\theta)=r)=1$ such that $\al=M\theta$. So $\B\subseteq\tB$ and $\tB\in\b_r(\A)$. Now suppose there exists some $\C\in\b_r(\A)$ such that $\tB\subseteq\C$. Then $\1_G\circ\A=\1_G\circ\tB\subseteq\1_G\circ\C\subseteq\1_G\circ\A$, so $\1_G\circ\tB=\1_G\circ\C$. For $s=r+1\ldots n$ and $F=\phi_s(X)$, we define $\A^s=\1_F\circ\A=\A(\1_F\bullet X)$ with $\1_F \bullet X$ satisfying the property $(\gamma)$ on $F$, $\tB^s=\1_F\circ\tB=\A(\1_F\theta\bullet X)$ and $\C^s=\1_F\circ\C$. We deduce that $\tB^s$ is maximal in $\b_r(\A^s)$ and $\C^s\in\b_r(\A^s)$, then $\tB^s=\C^s$ and therefore $\1_F\circ\tB=\1_{F}\circ\C$, which means that $\tB=\C$.

\end{proof}


\section{Complementarity and plug-in vector space.}

\subsection{Complementarity.}

As we know, in any finite dimensional vector space $E$, we can associate an orthogonal vector space $N^\perp$ to any vector subspace $N$ in $E$ which satisfies $E=N\oplus N^\perp$ and $dim(E)=dim(N)+dim(N^\perp)$. Similarly we investigate the notions of complementarity and strict complementarity of a element $\B\in\b_r(\A)$ with $\A\in\b_n$.

\begin{defi}\label{d4} Let $\A\in\b_n$ and $\B\in\b_r(\A)$ for some integers $r\leq n$.
\begin{enumerate}
\item We say that $\B$ has a complement set in $\A$ if there exists a minimal set $\C\in\b$ such that $\A=\B+\C$.
\item We say that $\B$ has a strict complement set in $\A$ if $\B$ has a complement set $\C$ in $\A$ such that $dim(\A)=dim(\B)+dim(\C)$.
\end{enumerate}
\end{defi}

We show the existence of a complement set and a strict complement set for $\B\in\b_r(\A)$. We start by the case where $\A=\A(X)$ and $X$ satisfies the property $(\gamma)$.

\begin{thm}\label{t3}Let $\A=\A(X)\in\b_n$ and $\B=\A(Y)\in\b_r(\A)$ for some integers $r\leq n$. Suppose $X\in\cc_n$ satisfies the property $(\gamma)$, then
\begin{enumerate}
\item $\B$ has a complement set in $\A$.
\item $\B$ has a strict complement set in $\A$ if and only if $\B$ is maximal in $\b_r(\A)$.

\end{enumerate}
\end{thm}

\begin{proof}
(1) Thanks to Proposition \ref{p11}, we get that $\B=\A(Y)=\oplus_{i=1}^r\,\1_{\phi_i(Y)}\circ\A(U^{1:i})$ where the process $U$ is the $\phi(Y)$-arrangement of $Y$ and $U^{1:i}=(U^1,\ldots,U^i)$. There exists then for each $i=1\ldots r$, some $\gamma^i\in\m_{i,n}(\cp)$ with $\tp(rank(\gamma^i)=i)=1$ such that $\B=\oplus_{i=0}^r\1_{\phi_i(Y)}\circ\A(\gamma^i\bullet X)$. Let for each $i=1\ldots r$, some $\delta^i\in\m_{n-i,n}(\cp)$ with $\tp(rank(\delta^i)=n-i)=1$ and orthogonal to $\gamma^i$. Then for $\C:=\oplus_{i=1}^r\1_{\phi_i(Y)}\circ\A(\delta^i\bullet X)$, we have $\B+\C=\oplus_{i=1}^r\1_{\phi_i(Y)}\circ\A(\gamma^i\bullet X)+\oplus_{i=1}^r\1_{\phi_i(Y)}\circ\A(\delta^i\bullet X)=\A(K\bullet X)$ with $K:=\oplus_{i=1}^r\1_{\phi_i(Y)}\left(\begin{array}{c}
\gamma^i \\
\delta^i
\end{array}\right)$ is of full rank and then invertible, so $\A(K\bullet X)=\A$. Now suppose $\A=\B+\D$ for some $\D\in\b(\C)$, then $\D=\A(\theta\bullet X)$ for some $\theta\in\m_{n,n}(\cp)$. We remark that for each $i=1\ldots n$, the vector space $span(\1_{\phi_i(Y)}\delta^i)$ is the orthogonal vector space of $span(\1_{\phi_i(Y)}\gamma^i)$ and that $\1_{\phi_i(Y)}\R^n=span(\1_{\phi_i(Y)}\gamma^i)+span(\1_{\phi_i(Y)}\theta)$ with $span(\1_{\phi_i(Y)}\theta)\subseteq span(\1_{\phi_i(Y)}\delta^i)$. Then $span(\1_{\phi_i(Y)}\delta^i)=span(\1_{\phi_i(Y)}\theta)$ and therefore $\C=\D$.

(2) Suppose $\B$ is maximal in $\b_r(\A)$, then the set $\C$ defined above is given by $\C=\A(\delta^r\bullet X)$ and therefore $\C$ has exactly dimension $n-r$. Conversely suppose that $\B$ has a strict complement set $\C$ in $\A$ and suppose that $\1_F\circ\B\in\b_s$ for some $F\in\cp$ with $s<r$. Then $\1_F\circ\A=\1_F\circ\B+\1_F\circ\C$ and therefore $dim(\1_F\circ\A)\leq dim(\1_F\circ\B)+dim(\1_F\circ\C)=s+n-r<n$, which is a contradiction.

\end{proof}

\begin{rem}\label{ryy1}We remark that there is more than one complement set for any $\B\in\b(\A)$. In the Brownian setting with $W=(W^1,W^2)$, we define $\A=\A(W^1,W^2)$ and $\B=\A(W^1)$. Then any set of the form $\A(a W^1+W^2)$ for a scalar $a$, is a complement set of $\B$ in $\A$.

\end{rem}

The general version of Theorem \ref{t3} is given below.

\begin{thm}\label{tt3}Let $\A=\A(X)\in\b_n$ and $\B=\A(Y)\in\b_r(\A)$ for some integers $r\leq n$. Then
\begin{enumerate}
\item $\B$ has a complement set in $\A$.
\item $\B$ has a strict complement set in $\A$ if and only if for all $i=n-r+1\ldots n$ and $j=0\ldots i+r-n-1$, we have that $\tp(\phi_i(X)\cap\phi_j(Y))=0$.

\end{enumerate}
\end{thm}

\begin{proof}
(1) For $i=0\ldots n$, we define $F^{i}=\phi_i(X)$ and $\A^{i}=\1_{F^{i}}\circ\A$. So the process $\1_{F^{i}}\bullet X$ satisfies the property $(\gamma)$ on $F^i$. We apply assertion (1) in Theorem \ref{t3} for $\B^{i}:=\1_{F^{i}}\circ \B$ as an element in $\b(\A^{i})$, and obtain that there exists a minimal set $\C^{i}$ in $\b(\A^{i})$ such that $\A^{i}=\B^{i}+\C^{i}$. We deduce that $\A=\B+\C$ with $\C=\oplus_{i}\1_{F^{i}}\circ\C^{i}$. Suppose $\A=\B+\D$ for some $\D\in\b(\C)$, then $\A^{i}=\B^{i}+\1_{F^{i}}\circ \D$ with $\1_{F^{i}}\circ \D\subseteq \1_{F^{i}}\circ \C=\C^i$. From the minimality of $\C^{i}$ we conclude that $\C^{i}=\1_{F^{i}}\circ \D$ and therefore $\C=\D$.

(2) For the direct implication, we suppose that $\tp(\phi_i(X)\cap\phi_j(Y))>0$ for some $i=n-r+1\ldots n$ and $j=0\ldots i+r-n-1$. Then $\1_{F^{ij}}\circ \A=\1_{F^{ij}}\circ \B+\1_{F^{ij}}\circ\C$ for $F^{ij}=\phi_i(X)\cap\phi_j(Y)$. We deduce that $i=dim(\1_{F^{ij}}\circ \A)\leq dim(\1_{F^{ij}}\circ \B)+dim(\1_{F^{ij}}\circ\C)\leq j+n-r$, which means that $i-j\leq n-r$ but $i-j>n-r$ which is a contradiction.

Conversely we define $I=\{(i,j)\in\{1,\ldots,n\}\times\{0,\ldots,r\}:\;0\leq i-j\leq n-r\}$, $J=\{(i,j)\in I:\;\tp(\phi_i(X)\cap\phi_j(Y))>0\}$ and for any $(i,j)\in J$ we define $F^{ij}=\phi_i(X)\cap\phi_j(Y)$, $\A^{ij}:=\1_{F^{ij}}\circ \A=\A(\1_{F^{ij}}\bullet X)$ and $\B^{ij}:=\1_{F^{ij}}\circ \B$, with $\1_{F^{ij}}\bullet X$ satisfying the property $(\gamma)$ on $F^{ij}$. We remark that $\B^{ij}$ is maximal in $\b_j(\A^{ij})$, so by Theorem \ref{t3}, we conclude that there exists some $\C^{ij}\in\b_{i-j}(\A^{ij})$ such that $\A^{ij}=\B^{ij}+\C^{ij}$ and then $\A=\B+\C$ with $\C=\oplus_{(i,j)\in J}\1_{F^{ij}}\circ\C^{ij}\in\b_{n-r}(\A)$.

\end{proof}

\begin{rem}In Theorem \ref{tt3} and under the condition that $X\in\cc_n$ satisfies the property $(\gamma)$, we have that the strict completeness is equivalent to maximality of $\B$ in $\b_r(\A)$, since $\tp(\phi_i(X))=0$ for all $i<n$ and then $\tp(\phi_j(Y))=0$ for all $j<r$, which means that $\B$ is maximal in $\b_r(\A)$. In the general case, we show next that maximality is a sufficient condition.
\end{rem}

\begin{prop}\label{xp13}Let $\A\in\b_n$ and $\B\in\b_r(\A)$ with $r\leq n$. Suppose that $\B$ is maximal in $\b_r(\A)$, then $\B$ admits a strict complement in $\A$.
\end{prop}

\begin{proof}
By applying Proposition \ref{xp12}, we have $\B=\1_G\circ\A+\1_{G^{c}}\circ\A(\theta\bullet X)$ for some $\theta\in\m_{r,n}(\cp)$ and $G=\cup_{i=0}^r\,\phi_i(X)$. For $j=r+1\ldots n$ we define $\ttheta^j\in\m_{j-r,n}(\cp)$ with rank $j-r$, to be orthogonal to $\theta$ and set $\ttheta=\oplus_{j=r+1}^n\,\1_{\phi_j(X)}\ttheta^j$. Then $\C=\1_{G^{c}}\circ\A(\ttheta\bullet X)\in\b_{n-r}(\A)$ and $\B+\C=\A(K\bullet X)$ where $K:=\1_G\,I_n+\1_{G^{c}}\left(\begin{array}{c}
\theta \\
\ttheta
\end{array}\right)$ is of full rank and then invertible, where $I_n$ is the identity matrix, so $\B+\C=\A$.

\end{proof}


\subsection{The plug-in vector space.}

For an integer $n$ and $X\in\cc_n$, we define the plug-in vector space:
$$
\v(X)=\overline{\{\al\in\Linf(\cp;\R^n):\;\al\bullet X\in m(\q^X)\}}^0,
$$
where the closure is taken in the sense of convergence in measure, and denote by $\w_n$, the set of all $\cp$-measurable mappings with vector space values in $\R^n$. We will explore the link between the two sets $\cc_n$ and $\w_n$.

\begin{prop}\label{pp2}Let $X\in\cc_n$, then there exists a $\cp$-measurable mapping $W(X)\in\w_n$ such that $\v(X)=\{\al\in\L^0(\cp;\R^n):\;\al\in W(X)\;\mbox{a.s.}\}$, with its random algebraic dimension $\dd=\dd(X)$ and its algebraic basis $f_X=(f^1_X,\ldots,f^{\dd}_X)$. Moreover $\A(X)=\A(g_X\bullet X)$ with $g_X=f_X/(1+|f_X|)$, $\dd(X)=\oplus_{k=0}^n\,k\1_{\phi_k(X)}$ and the rank of $X$ is the essential supremum of $\dd(X)$.

\end{prop}

\begin{proof}The vector space $\v(X)$ is closed in $\L^0(\cp,\R^n)$ and closed under multiplication by bounded positive $\cp$-measurable random variables, so thanks to Lemma A.4 in \cite{Sch1} and Lemma 2.5 in \cite{Sch2}, there exists a vector space valued $\cp$-measurable mapping $W(X)$ such that $\v(X)=\{\al\in\L^0(\cp;\R^n):\;\al\in W(X)\;\mbox{a.s.}\}$. Let $\dd(X)$ be the algebraic dimension of $W(X)$ and $f_X=(f^1_X,\ldots,f^{\dd}_X)$ its algebraic basis, then $\A(g_X\bullet X)\subseteq\A(X)$. For the converse inclusion, let $\al\in \v(X)\cap\Linf(\cp;\R^n)$, then $\al=M g_X$ for some predictable matrix process $M$ and therefore $\al\bullet X=Mg_X\bullet X$ and so $\A(X)\subseteq\A(g_X\bullet X)$. Now for the algebraic dimension of $W(X)$, we shall show that $F_s:=\cup_{k=s}^n\phi_k(X)=\{\dd(X)\geq s\}=:K_s$ for $s=0\ldots n$. We suppose by absurd that $\tp(F_s\cap (K_s)^{c})>0$ for some $s$, then $dim(\1_F\circ\A(X))\geq s$ for $F:=F_s\cap (K_s)^{c}$ and $dim(\1_F\circ\A(X))=dim(\1_F\circ\A(g_X\bullet X))\leq s-1$, which is a contradiction. We suppose again by absurd that $\tp(K_s\cap (F_s)^{c})>0$ for some $s$, then for $K:=K_s\cap (F_s)^{c}$, we have $\1_K\circ\A(X)=\1_K\circ\A(U)$ for some $U\in\cc_{s-1}$ and then $\1_K \dd(X)=\1_K \dd(U)\leq s-1$ but $\1_K \dd(X)\geq s$. This is a contradiction. We deduce that $\phi_n(X)=F_n=K_n=\{\dd(X)= n\}$ and that for all $s=0\ldots n-1$ we have $\phi_s(X)=F_{s}\cap(F_{s+1})^{c}=K_{s}\cap(K_{s+1})^{c}=\{\dd(X)=s\}$.
\end{proof}

\begin{lem}\label{l2}Let $X\in\cc_n$ and $F\in\cp$. Then $\dd(\1_{F}\bullet X)=\1_{F}\dd(X)$.

\end{lem}
\begin{proof}We have that $\1_{F\cap\phi_k(X)}\circ\A=\1_{\phi_k(X)}\circ\A(\1_F\bullet X)$ for all $k=0\ldots n$, then $F\cap\phi_k(X)=\phi_k(\1_F\bullet X)$ and therefore
$$
\1_F\dd(X)=\oplus_{k=0}^n\,k\1_{F\cap\phi_k(X)}=\oplus_{k=0}^n\,k\1_{\phi_k(\1_F\bullet X)}=\dd(\1_F\bullet X).
$$
\end{proof}

\begin{prop}\label{pe1}Let $X\in\cc_n$ and $Y\in\cc_r$. Let the following assertions:
\begin{enumerate}
\item $ \A(X)=\A(Y)$.
\item $\dd(X)=\dd(Y)$.
\end{enumerate}
Then $(1)\Rightarrow(2)$. If moreover $\A(X)\subseteq\A(Y)$, then $(1)\Leftrightarrow(2)$.
\end{prop}

\begin{proof} $(1)\Rightarrow(2)$ Suppose $\A(X)=\A(Y)$, then $r(X)=r(Y)$. We deduce then that $W(X)=W(Y)$ and therefore $\dd(X)=\dd(Y)$.

$(2)\Rightarrow(1)$ Now suppose further $\A(X)\subseteq\A(Y)$ and $\dd(X)=\dd(Y)$, then $r(X)=r(Y)=:s$, $\phi(X)=\phi(Y)=:F$ and for all $k=1\ldots s$ we have $\1_{F^k}\bullet U^{1:k}_X=\1_{F^k}\theta^k\bullet U^{1:k}_Y$ for some $\theta^k\in\m_{k,k}(\cp)$ where $U_X$ and $U_Y$ are respectively the $F$-arrangement processes of $X$ and $Y$. Suppose $\tp(G)>0$ with $G:=\{rank(\theta^k)<k;F^k\}$, then $r(\1_{G}\bullet U_X^{1:k})<k$, which is a contradiction. Therefore $\tp(G)=0$ and $\1_{F^k}\bullet U^{1:k}_Y=\1_{F^k}(\theta^k)^{-1}\bullet U^{1:k}_X$, which means that $\1_{F^k}\circ\A(Y)=\1_{F^k}\circ\A(X)$ and then $\A(Y)=\A(X)$.

\end{proof}

The converse of Proposition \ref{pp2} is given below.

\begin{prop}\label{xp10}Let $W\in\w_n$ with its random algebraic dimension $\dd$ and its algebraic basis $f=(f^1,\ldots,f^{\dd})$. Let us suppose that there exists some $R\in\cc_n$ satisfying the property $(\gamma)$ and define $X=g\bullet R$ for $g=f/(1+|f|)$, then $W=W(X)$.

\end{prop}

\begin{proof}Thanks to Propositions \ref{p6} and \ref{pp2}, we have $\{\dd(X)=s\}=\phi_s(X)=\{rank(f)=s\}$ for all $s=0\ldots n$ and since $\dd=rank(f)$, then $\dd=\dd(X)$. In order to show that $W=W(X)$, we remark that $f\in W(X)$, then $W\subseteq W(X)$ and therefore $W=W(X)$.

\end{proof}

As it was mentioned before, there is no uniqueness of a complement set. Next we use the tool of the plug-in vector space to construct a complement set, which satisfies the property of orthogonality and by consequent unique.

\begin{thm}\label{ttt5}Let $X\in\cc_n$ and $Y\in\cc_r$ satisfying $\A(Y)\subseteq\A(X)$. Then there exists $s\leq n$, $\theta\in\m_{r,n}(\cp)$ and $\theta'\in\m_{s,n}(\cp)$ such that:
\begin{enumerate}
\item $Y':=\theta'\bullet X\in\cc_s$ is orthogonal to $Y$.
\item $W(X)=W(Y).\theta+W(Y').\theta'$.
\item $\dd(X)=\dd(Y)+\dd(Y')$.
 \end{enumerate}

\end{thm}

\begin{proof}We adopt the notation $X\eqf Y$ for two processes $X,Y$ and $F\in\cp$ when $\1_F\bullet X=\1_F\bullet Y$. We start first by considering the special case where the two processes $X$ and $Y$ satisfy the assumption $(\gamma)$. Then thanks to assertion (2) in Theorem \ref{t3}, there exists some $\theta'\in\m_{n-r,n}(\cp)$ such that $Y':=\theta'\bullet X\in\cc_{n-r}$ and $\A(X)=\A(Y)+\A(Y')$. So $\A(Y')$ admits a strict complement set in $\A$ and therefore by applying Theorem \ref{t3} again we deduce that $Y'$ satisfies the assumption $(\gamma)$ and then $\dd(Y')=n-r=\dd(X)-\dd(Y)$. Next we show that $\V(X)= \V(Y).\theta+\V(Y').\theta'$. Let $\al\in \V(X)\cap\Linf(\cp;\R^n)$ and since $\Ml(X)=\Ml(Y,Y')$, we get $\al\bullet X=\beta\bullet Y+\beta'\bullet Y'=(\beta\theta+\beta'\theta')\bullet X$ with $\beta\in \V(Y)$ and $\beta'\in \V(Y')$. From the property $(\gamma)$ of $X$ we obtain that $\al=\beta\theta+\beta'\theta'$. For the converse inclusion let $\al\in \V(Y)$, then $\al\theta\bullet X=\al\bullet Y\in m(\q^Y)$ with $m(\q^Y)\subseteq m(\q^X)$, then $\al\theta\in \V(X)$ and since the set $\V(Y).\theta+\V(Y').\theta'$ is closed in $\L^0(\cp;\R^n)$, then $\V(X)=\V(Y).\theta+\V(Y').\theta'$. We conclude that $\V(X)=\{\al\in\L^0(\cp;\R^n):\;\al\in W(Y).\theta+W(Y').\theta'\}$ and deduce that $W(X)=W(Y).\theta+W(Y').\theta'$. To show orthogonality we have that $K:=K_2\cap K_1^{\perp}\in\w_n$ with $K_1=W(Y).\theta$ and $K_2=W(Y').\theta'$. Thanks to Theorem \ref{xp10}, there exists some $u\leq n$ and $\al\in\m_{u,n}(\cp)$ such that $\al\bullet X\in\cc_u$, $K=W(\al\bullet X)$ and since $W(X)=K_1+K$, then $\A=\B+\A(\al\bullet X)$ with $\A(\al\bullet X)\subseteq\A(Y')$, so $\A(Y')=A(\al\bullet X)$ from the minimality of $\A(Y')$ and therefore $K_2=K$ which means that $K_1\perp K_2$.

Now for the general case, we show first that the quadratic variation process $[X]$ of $X$ can be written as $[X]=C\bullet K$ for a positive definite matrix $C\in\m_{n,n}(\cp)$ and a predictable increasing process $K$. Indeed we define the process $K=\sum_{i,j}|[X^i,X^j]|$, where $|[X^i,X^j]|$ is the total variation process of $[X^i,X^j]$. Then each process $[X^i,X^j]$ is absolutely continuous w.r.t $K$ and therefore there exists $C_{ij}$ such that $[X^i,X^j]=C_{ij}\bullet K$ and then $[X]=C\bullet K$. To show that $C$ is positive definite, let $\al\in\m_{1,n}(\cp)$, $\lam\in\Linf_+$ and constat that $\te\left(\lam\al C\al\bullet K\right)=\te|\left(\sqrt{\lam}\al \bullet X\right)|^2\ge 0$, therefore $\al C\al\geq 0$.

For the remainder of the proof, we suppose w.l.g. that $X=U_X$ and $Y=U_Y$. We define $\Lam:=\{(i,j)\in\{1,\ldots,n\}\times\{1,\ldots,r\}:\;\tp(\phi_i(X)\cap\phi_j(Y))>0\}$. Fix $(i,j)\in\Lam$ and $F=F^{ij}:=\phi_i(X)\cap\phi_j(Y)$, so we remark that $\1_{F}\circ\A(Y^{1:j})=\1_{F}\circ\A(Y)\subseteq\1_{F}\circ\A(X)=\1_{F}\circ\A(X^{1:i})$ and therefore $Y^{1:j}\eqf\beta^{ij}\bullet X^{1:i}$, $Y^{1:j}\eqf M^j_Y\bullet Y$, $X^{1:i}\eqf M^i_X\bullet X$, $Y\eqf N^j_Y\bullet Y^{1:j}$ and $X\eqf N^i_X\bullet X^{1:i}$ for some $\beta^{ij}\in\m_{j,i}(\cp)$, $M^j_Y\in\m_{j,r}(\cp)$, $M^i_X\in\m_{i,n}(\cp)$, $N^j_Y\in\m_{r,j}(\cp)$ and $N^i_X\in\m_{n,i}(\cp)$. We also have that
$[X^{1:i}]=H^i\bullet K$ with $H^i:=M^i_X C M^i_X$, we shall show that $H^i$ is invertible. Let $\al\in\m_{1,i}(\cp)$ such that $\al H^i=0$, then $\e(\al\bullet X^{1:i})^2=\e(\al H^i\al\bullet K)=0$ and therefore $\al=0$ on $F$ since the process $X^{1:i}$ satisfies the property $(\gamma)$ on $F$. Since $H^i$ is invertible and positive definite, then there exists $R^i\in\m_{i,i}(\cp)$ such that $H^i=R^iR^i$ and $R^i$ is invertible. We conclude that the process $U^i=D^i\bullet X^{1:i}$ satisfies the property $(\gamma)$ on $F$, where $D^i$ is the inverse matrix of $R^i$, and since $Y^{1:j}\eqf \beta^{ij}R^i\bullet U^i$, we deduce from the first case, that there exists some $\al^{ij}\in\m_{n,i}(\cp)$ such that $Z^{ij}:=\al^{ij}\bullet U^i\in\cc_{n}$, $\1_{F}W(U^i)=\1_{F}W(Y^{1:j}).\beta^{ij}R^i+\1_{F}W(Z^{ij}).\al^{ij}$ and the two vector spaces $\1_{F}W(Y^{1:j}).\beta^{ij}R^i$ and $\1_{F}W(Z^{ij}).\al^{ij}$ are algebraically orthogonal. We remark that $\1_{F}W(U^{i})=\1_{F}W(X^{1:i})R^i=\1_{F}W(X)N^i_X R^i$, that $\1_{F}W(Y^{1:j})=\1_{F}W(Y)M^j_Y$ and that $\1_G W(K)=W(\1_G\bullet K)$ for any $G\in\cp$ and $K\in\cc_n$, we obtain then that $\1_{F}W(X)=\1_{F}W(Y).\theta+\1_{F}W(Y').\theta'$, where $Y'=\sum_{(i,j)\in\Lam}\1_{F^{ij}}\bullet Z^{ij}$, $\theta=\sum_{(i,j)\in\Lam}\1_{F^{ij}}M^j_Y\beta^{ij}M^i_X$ and $\theta'=\sum_{(i,j)\in\Lam}\1_{F^{ij}}\al^{ij}D^i M^i_X$. We take the sum over the set $\Lam$ and obtain that $W(X)=W(Y).\theta+W(Y').\theta'$ with $Y=\theta\bullet X$, $Y'=\theta'\bullet X$. We constat also, by using Lemma \ref{l2}, that $\dd(X)=\sum_{(i,j)\in\Lam}\1_{F^{ij}}\dd(X)=\sum_{(i,j)\in\Lam}\1_{F^{ij}}\dd(X^{1:i})=\sum_{(i,j)\in\Lam}\1_{F^{ij}}\left(\dd(Y^{1:j})+\dd(Z^{ij})\right)$$
=\sum_{(i,j)\in\Lam}\1_{F^{ij}}\left(\dd(Y)+\dd(Y')\right)=\dd(Y)+\dd(Y').
$

It remains to show that the two processes $Y$ and $Y'$ are orthogonal. Let $f\in W(Y)$ and $g\in W(Z)$, then
$$
(f.\theta)C(g.\theta')=\sum_{(i,j)\in\Lam}\1_{F^{ij}}(f.M^j_Y\beta^{ij}M^i_X)C(g.\al^{ij}D^i M^i_X)
$$
$$
=\sum_{(i,j)\in\Lam}\1_{F^{ij}}(f.M^j_Y\beta^{ij})M^i_X C M^i_X D^i(g.\al^{ij})
$$
$$
=\sum_{(i,j)\in\Lam}\1_{F^{ij}}(f.M^j_Y\beta^{ij})H^i D^i(g.\al^{ij})
$$
$$
=\sum_{(i,j)\in\Lam}\1_{F^{ij}}(f.M^j_Y\beta^{ij}R^i)(g.\al^{ij})=0.$$
So $\theta.C.\theta'=0$ and therefore $[Y,Y']=\theta.C.\theta'\bullet K=0$.

\end{proof}


\section{Further results.}

\subsection{The orthogonal complement set.}
For $X\in\cc_{n}$ and $\A=\A(X)$, we denote $\dd(\A):=\dd(X)$, $U_X$ is the $\phi(X)$-arrangement of the process $X$ and $\k(\A)=\A\cap-\A$. For $f\in\k(\A)$, we denote $M^f$ the associated $\q$-martingale, defined by $M^f_t=\eq(f|\F_t)$ for any $\bq\in\q$. We remark that $\k(\A)$ is the set of elements $M_T$, where $M$ is a bounded $\q$-martingale with $M_0=0$, and shall write $f\perp g$ for $f,g\in\k(\A)$ if $M^f\perp M^g$.

The following result will be a basic tool in showing some of the next results.
\begin{prop}\label{pop}Let suppose $(X,Y)\in\cc_{n+r}$ with $\A=\A(X)$ and $\B=\A(Y)$. Then
\begin{enumerate}
\item[(i)] $\A$ is the weak star closure of $\k(\A)+\Linf_-$.
\item[(ii)] $\k(\A)$ is a vector space.
\item[(iii)] $\B\subseteq\A$ iff $\k(\B)\subseteq\k(\A)$.
\item[(iv)] $\k(\A+\B)=\k(\A)+\k(\B)$.
\item[(v)] $\k(\A)\cap\k(\B)=\{0\}$ if $X$ is orthogonal to $Y$.
\end{enumerate}
\end{prop}

\begin{proof} $(i)$ The inclusion $\overline{\k(\A)+\Linf_-}^*\subseteq\A$ is trivial. Let us show the direct one. Let $h\in\A$, then thanks to Kramkov's theorem in \cite{Kram}, there exists a local $\q$-martingale $M$ and an increasing process $C$ such that $h=\ce(h)+M_T-C_T$ where $\ce$ is the sublinear expectation operator associated to $\q$. Therefore $h^n:=\ce(h)+M^n_T-C^n_T\in\k(\A)+\Linf_-$ where $M^n$ and $C^n$ are the truncated versions of $M$ and $C$ respectively. Since $h^n$ converges to $h$, then $h\in \overline{\k(\A)+\Linf_-}^*$.

Assertions $(ii)$ and $(iii)$ are trivial.

$(iv)$ The inclusion $\k(\A)+\k(\B)\subseteq\k(\A+\B)$ is a consequence of assertion $(iii)$. For the other inclusion, let $h\in\k(\A+\B)$, then $h=f+g$ with $f\in\A$ and $g\in\B$. There exists then a local $\q$-martingale $M$ and a local $\q^{\B}$-martingale $N$ such that $f\leq M_T$ and $g\leq N_T$. So $h\leq M_T+N_T$ and since $\bq(h)=0$ and $\bq(M_T+N_T)\leq 0$ for some $\bq\in \Ml^e(X,Y)$, then $h=M_T+N_T$. Therefore $0=(M_T-f)+(N_T-g)$, we conclude that $f=M_T\in\k(\A)$ and $g=N_T\in\k(\B)$, so $h\in\k(\A)+\k(\B)$.

$(v)$ Let $h\in\k(\A)\cap\k(\B)$. Then $h=M^f_T=M^g_T$ for some $f\in\k(\A)$ and $g\in\k(\B)$. Since there exists $\bq\in\Ml^e(X,Y)$, then $M^f=M^g$. Since $M^f=\al^f\bullet X$ and $M^g=\al^g\bullet Y$ for two vector valued predictable processes $\al^f$ and $\al^g$, we conclude that $M^f=0$ and therefore $h=0$.
\end{proof}

Now we state the existence of a complement set with the orthogonality property.

\begin{thm}\label{tt2}Let $\A=\A(X)\in\b_n$ and $\B=\A(Y)\in\b_r(\A)$. Then three exists $s\leq n$ and a unique $Y'\in\cc_s$ up to a class of equivalence, such that:
\begin{enumerate}
\item $Y'$ is orthogonal to $Y$.
\item $\A(Y')$ is a complement of $\B$ in $\A$.
\item $\dd(\A(Y'))=\dd(\A)-\dd(\B)$.
\end{enumerate}

\end{thm}
\begin{proof}Thanks to Theorem \ref{ttt5}, there exists $s\leq n$ and $Y'\in\cc_s$ orthogonal to $Y$ such that $\dd(X)=\dd(Y)+\dd(Y')$ and $\A=\B+\A(Y')$. Now let $\C\in\b(\A(Y'))$ such that $\A=\B+\C$, so $\dd(\C)\geq \dd(\A)-\dd(\B)=\dd(\A(Y'))$ and therefore $\C=\A(Y')$ thanks to Proposition \ref{pe1}. To show uniqueness of $\A(Y')$, we consider $\C=\A(U)\in\b(\A)$ such that $U$ satisfies the same assertions as $Y'$. Let $h\in\k(\C)$ and since $\C\subseteq\A=\B+\A(Y')$, then by applying assertion $(iv)$ in Proposition \ref{pop}, there exists $f\in\k(\B)$ and $g\in\k(\A(Y'))$ such that $h=f+g$. We deduce that $h-g=f$ with $h-g\perp f$, so $h=g\in\A(Y')$ and therefore $\C\subseteq\A(Y')$. In the same way we show the other inclusion and conclude that $\C=\A(Y')$.

\end{proof}
We shall write $\B\perp\C$ if $\B=\A(X)$, $\C=\A(Y)$ and $X\perp Y$.

\begin{defi}\label{d8}The set $\A(Y')$ in Theorem \ref{tt2} will be called the orthogonal complement of $\B$ in $\A$ and will be denoted by $c(\A,\B)$. We write in this case $\A=\B\oplus c(\A,\B)$.
\end{defi}

Some consequences of Theorem \ref{tt2} are given below. We start by stating a general formula for the random dimension of the sum of two subsets of $\b(\A)$ with $\A\in\b_n$.

\begin{thm}\label{ppp}Let $\A\in\b_n$ and $\B,\C\in\b(\A)$. Then $\dd(\B+\C)=\dd(\B)+\dd(\C)-\dd(\B\cap \C)$.
\end{thm}
\begin{proof}Let $X,Y\in\cc:=\cup_{k\geq 1}\cc_k$ such that $\B=\A(X)$ and $\C=\A(Y)$. We suppose first that $\k(\B)\cap\k(\C)=\{0\}$. We shall show that for all $i\leq r(X)$ and $j\leq r(Y)$, the process $(\1_F\bullet U_X^{1:i},\1_F\bullet U_Y^{1:j})$ satisfies the property $(\gamma)$ for $F=\phi_i(X)\cap\phi_j(Y)$. Let two vector valued predictable processes $\al$ and $\beta$ such that $\1_F\al\bullet U_X^{1:i}+\1_F\beta\bullet U_Y^{1:j}=0$. Then by truncating if necessary we get that $f:=\left(\1_F\al\bullet U_X^{1:i}\right)_T\in \k(\B)\cap\k(\C)$ and therefore $f=0$ which means that $\1_F\al\bullet U_X^{1:i}=\1_F\beta\bullet U_Y^{1:j}=0$ and by consequent $\al \1_F=\beta\1_F=0$. We deduce that $\1_F\dd(X)+\1_F\dd(Y)=\1_F\dd(U_X^{1:i})+\1_F\dd(U_Y^{1:j})=\1_F\dd(U_X^{1:i},U_Y^{1:j})=\1_F\dd(X,Y)$, therefore $\dd(X)+\dd(Y)=\dd(X,Y)$.

Now for the general case we write $\B+\C=\B+\D$ with $\D=c(\C,\C\cap\B)$. We shall verify that $\k(\B)\cap\k(\D)=\{0\}$. Let $f\in \k(\B)\cap\k(\D)$, then $f\in\B\cap\C$ and $f\perp\B\cap\C$, so $f=0$. We apply the first case and deduce that $\dd(\B)+\dd(\D)=\dd(\B+\D)=\dd(\B+\C)$ and since $\dd(\D)=\dd(\C)-\dd(\B\cap\C)$, we conclude the result.
\end{proof}

A consequence of Theorem \ref{ppp}, characterizing the complementarity property, is given below.

\begin{cor}\label{cccc}Let $\A\in\b_n$ and $\B,\C\in\b(\A)$. Then $\C$ is a complement of $\B$ in $\A$ iff $\A=\B+\C$ and $\B\cap\C=\Linf_-$.
\end{cor}
\begin{proof}For the direct implication, let $\D$ be the orthogonal complement of $\B\cap\C$ in $\C$. Then $\A=\B+\D$ and $\D\subseteq\C$ and since $\C$ is minimal, we deduce that $\D=\C$ and therefore $\B\cap\C=\Linf_-$. For the converse let $\D\in\b$ satisfying $\A=\B+\D$ and $\D\subseteq\C$, then $\D\cap\B=\Linf_-$ and $\dd(\D)=\dd(\A)-\dd(\B)=\dd(\C)$ and therefore $\C=\D$ thanks to Proposition \ref{pe1}.
\end{proof}

The second consequence is a generalization of the well known Kunita-Watanabe martingale decomposition theorem.

\begin{thm}\label{to1}Let suppose that $(X,Y)$ is an $\R^{d}\times\R^k$-valued local martingale. Then there exists a predictable matrix valued process $\al$ and a vector valued local martingale $L$, orthogonal to $X$ such that $Y=\al\bullet X+L$.
\end{thm}

\begin{proof} First $(X,Y)\in\cc_{d+k}$ and $\A(X)\subseteq\A(X,Y)$, so by Theorem \ref{tt2}, there exists some $Y'\in\cc$, orthogonal to $X$ such that $\A(X,Y)=\A(X,Y')$. Since $Y$ is a local $\Ml(X,Y)$-martingale, then it is a local $\Ml(X,Y')$-martingale. By Jacka's theorem in \cite{jacka}, there exists a predictable process $(\al,\beta)$ such that $Y=\al\bullet X+\beta\bullet Y'=\al\bullet X+L$ with $L:=\beta\bullet Y'\perp X$.

\end{proof}

The third consequence is a Gram-Schmidt type procedure applied to a process $X\in\cc_n$.

\begin{thm}\label{ccc5"}Let $X\in\cc_n$. Then there exists $Y=(Y^1,\ldots,Y^n)\in\cc_n$ such that each $Y^i$ and $Y^j$ are orthogonal for $i\neq j$ and $\A(X)=\A(Y)$.
\end{thm}

\begin{proof}We show by induction on $k=1\ldots n-1$ that there exists $Y^1,\ldots,Y^k\in\cc_1$ and $U_{k}\in\cc_{n-k}$ such that the processes $Y^1,\ldots,Y^k,U_{k}$ form an orthogonal family and $\A(X)=\A(Y^1,\ldots,Y^k,U_{k})$.

For $k=1$, we define $Y^1=X^1$ and by applying Theorem \ref{tt2}, there exists $s\leq n$ and $U_1\in\cc_s$, orthogonal to $Y^1$ such that $\A(X)=\A(Y^1,U_1)$. We claim that $s\leq n-1$ since $\1_{\phi_0(Y^1)}\dd(X)\leq n-1$ and $\dd(U_1)=\dd(X)-\dd(Y^1)=\1_{\phi_0(Y^1)}(\dd(X)-\dd(Y^1))+\1_{\phi_1(Y^1)}(\dd(X)-\dd(Y^1))\leq n-1$. We suppose now that the induction hypothesis is true until $k$, then there exists $Y^1,\ldots,Y^k\in\cc_1$ and $U_{k}\in\cc_{n-k}$ such that the processes $Y^1,\ldots,Y^k,U_{k}$ form an orthogonal family and $\A(X)=\A(Y^1,\ldots,Y^k,U_{k})$. By applying the case $k=1$ to the process $U_k$, there exists $Y^{k+1}\in\cc_1$ and $U_{k+1}\in\cc_{n-k-1}$ such that the processes $Y^{k+1}$ and $U_{k+1}$ are orthogonal and $\A(U_k)=\A(Y^{k+1},U_{k+1})$. So the processes $Y^1,\ldots,Y^{k+1},U_{k+1}$ form an orthogonal family and $\A(X)=\A(Y^1,\ldots,Y^{k+1},U_{k+1})$.

For $k=n-1$, there exists $Y^1,\ldots,Y^{n-1}\in\cc_1$ and $U_{n-1}\in\cc_{1}$ such that the processes $Y^1,\ldots,Y^{n-1},U_{n-1}$ form an orthogonal family and $\A(X)=\A(Y^1,\ldots,Y^{n-1},U_{n-1})$. We take $Y^n=U_{n-1}$ and obtain the result.
\end{proof}

Now for a fixed $\A\in\b_n$, we study the properties of the mapping $\B\in\b(\A)\rightarrow c(\A,\B)\in\b(\A)$.

\begin{prop}\label{p61}Let $\A\in\b_n$ and $\B,\C\in\b(\A)$. Then
 \begin{enumerate}
 \item $c(\A,c(\A,\B))=\B$ and therefore $c(\A,\B)=c(\A,\C)$ iff $\B=\C$.
 \item Suppose $\C\subseteq\B$, then $c(\A,\C)=c(\A,\B)\op c(\B,\C)$.
 \item Suppose $\C\subseteq c(\A,\B)$, then $c(\A,\C)=\B\op c(c(\A,\B),\C)$ and $c(\A,\B+\C)=c(c(\A,\B),\C)$.
 \item $\C\subseteq c(\A,\B)$ iff $\B\subseteq c(\A,\C)$ and in this case we have that $c(c(\A,\B),\C)=c(c(\A,\C),\B)$.
 \item $\B\subseteq\C$ iff $c(\A,\C)\subseteq c(\A,\B)$.
 \item $\C\perp\B$ iff $\C\subseteq c(\A,\B)$.
\item Suppose $\C\subseteq\B$, then $c(\B,\C)=\B\cap c(\A,\C)$.
\end{enumerate}

\end{prop}

\begin{proof}(1) We have $\B=c(\A,c(\A,\B))$ from the uniqueness of the orthogonal complement. Suppose now that $c(\A,\B)=c(\A,\C)$, then $\B=c(\A,c(\A,\B))=c(\A,c(\A,\C))=\C$.

(2) We have $\A=c(\A,\B)\oplus\B=c(\A,\B)\oplus c(\B,\C)\oplus\C$ and $\A=c(\A,\C)\oplus\C$, so by uniqueness we get $c(\A,\C)=c(\A,\B)\oplus
c(\B,\C)$.

(3) We replace in assertion (2) the set $\B$ by $c(\A,\B)$ and get the first statement. For the second one, we have $\A=c(\A,\B)\oplus\B=c(c(\A,\B),\C)\oplus\C\oplus\B$ and $\A=(\B+\C)\oplus c(\A,\B+\C)$, then by uniqueness $c(\A,\B+\C)=c(c(\A,\B),\C)$.

(4) Since there is a symmetry we shall show the direct implication. We apply assertion (3) and deduce that $\B\subseteq c(\A,\C)$ and that $c(c(\A,\B),\C)=c(\A,\B+\C)=c(\A,\C+\B)=c(c(\A,\C),\B)$. The converse implication is obtained by applying the direct implication again.

(5) Suppose $\B\subseteq\C$ and since $\C=c(\A,c(\A,\C))$, then by assertion (4) we have $c(\A,\C)\subseteq c(\A,\B)$.

(6) The inverse implication is trivial since $\B\perp c(\A,\B)$, let us show the direct one. Let $h\in\k(\C)$, we apply assertion $(iv)$ in Proposition \ref{pop} and obtain that $h=f+g$ with $f\in\k(\B)$ and $g\in\k(c(\A,\B))$. Then $h-g=f$ and $h-g\perp f$, therefore $h=g\in c(\A,\B)$.

(7) The direct inclusion is trivial. For the converse we have $c(\A,\C)=c(\A,\B)\oplus c(\B,\C)$ and $\D:=\B\cap c(\A,\C)\perp c(\A,\B)$, so $\D\subseteq c(\B,\C)$.

\end{proof}

Next we state De Morgan laws.

\begin{prop}\label{p8}Let $\A\in\b_n$ and $\B,\C\in\b(\A)$. Then
\begin{enumerate}
\item $c(\A,\B\cap\C)=c(\A,\B)+c(\A,\C)$.
\item $c(\A,\B+\C)=c(\A,\B)\cap c(\A,\C)$.
\end{enumerate}

\end{prop}
\begin{proof}(1) We suppose first that $\A=\B+\C$. We start by showing that $c(\A,\B)\cap c(\A,\C)=\Linf_-$. Let $f\in\k(c(\A,\B)\cap c(\A,\C))$, then $f\perp\B$ and $f\perp\C$ and therefore $f\perp\B+\C$, but $f\in\A=\B+\C$, so $f=0$. Now we have $\D:=c(\A,\B)+c(\A,\C)\subseteq c(\A,\B\cap\C)$ and $\dd(\D)=\dd(\A)-\dd(\B)+\dd(\A)-\dd(\C)=\dd(\A)-\dd(\B\cap\C)$ and by Proposition \ref{pe1} we get the result. For the general case, we have $c(\A,\B\cap\C)=c(\A,\B+\C)+c(\B+\C,\B\cap\C)=c(\A,\B+\C)+c(\B+\C,\B)+c(\B+\C,\C)=c(\A,\B)+c(\A,\C)$.

(2) We apply assertion (1) and obtain that $c(\A,c(\A,\B)\cap c(\A,\C))=c(\A,c(\A,\B))+c(\A,c(\A,\C))=\B+\C$, so $c(\A,\B)\cap c(\A,\C)=c(\A,\B+\C)$.

\end{proof}

For any $\A\in\b_n$ and $\B,\C\in\b(\A)$, we know from Theorem \ref{tt2} that there exists a unique orthogonal complement of $\B$ in $\B+\C$, denoted by $c(\B+\C,\B)$. In the following example we show that there is no direct relationship of inclusion between $c(\B+\C,\B)$ and $C$.

\begin{ex}\label{ex5} We consider the Brownian setting $W=(W^1,W^2)$ and define $\B=\A(W^1)$ and $\C=\A(W^1+W^2)$. Then $c(\B+\C,\B)=\A(W^2)$ and $\A(W^2)$ has no direct link to $\C$.
\end{ex}

By consequent we introduce the following definition.

\begin{defi}\label{d9}Let $\A\in\b_n$ and $\B,\C\in\b(\A)$. We say that the pair $(\B,\C)$ satisfies the star property if $c(\B+\C,\B)\subseteq\C$.

\end{defi}

Some equivalent statements of the star property are given below.

\begin{prop}\label{p7}Let $\A\in\b_n$ and $\B,\C\in\b(\A)$. Then the following assertions are equivalent:
\begin{enumerate}
\item the pair $(\B,\C)$ satisfies the star property.
\item $c(\B+\C,\B)=c(\C,\B\cap\C)$.
\item there exists $Y'\in\cc$ such that $\A(Y')$ is orthogonal to $\B$ and $\C=(\B\cap\C)\oplus\A(Y')$.
\item there exists $\D_1,\D_2\in\b(\A)$ such that $\D_1\subseteq\B$, $\D_2\perp\B$ and $\C=\D_1\oplus\D_2$.
\item $\B+\C=\B\oplus c(\C,\B\cap\C)$.
\item $\B+\C=\B\oplus \D$ for some $\D\in\b(\C)$.

\end{enumerate}

\end{prop}

\begin{proof}$(1)\Rightarrow (2)$ Let $\D:=c(\B+\C,\B)$ and $h\in\k(\D)$ and since $\D\subseteq\C$ and $\C=(\B\cap\C)\oplus c(\C,\B\cap\C)$, then $h=f+g$ with $f\in\k(\B\cap\C)$ and $g\in\k(c(\C,\B\cap\C))$. So $h-g=f$ with $h-g\perp f$, and then $h=g\in c(\C,\B\cap\C)$. Inversely Let $h\in\k(c(\C,\B\cap\C))$ and since $c(\C,\B\cap\C)\subseteq\C\subseteq\B+\C=\B\oplus\D$, then $h=f+g$ with $f\in\k(\B)$ and $g\in\k(\D)$. So $h\in\B\cap\C$, $h-g=f$ with $h-g\perp f$, then $h=g\in \D$.

$(2)\Rightarrow (3)$ Thanks to Theorem \ref{tt2}, there exists some $s\leq n$ and $Y'\in\cc_s$ such that $\A(Y')$ is orthogonal to $\B$ and $c(\B+\C,\B)=\A(Y')$. So $\C=(\B\cap\C)\oplus c(\C,\B\cap\C)=(\B\cap\C)\oplus \A(Y')$.

$(3)\Rightarrow (4)$ is trivial.

$(4)\Rightarrow (3)$ Suppose $\D_2=\A(Y')$. It suffices to show that $\D_1=\B\cap\C$. The direct inclusion is trivial, for the converse we have that $\B\cap\C\perp\D_2$, then $\B\cap\C\subseteq c(\C,\D_2)=\D_1$.

$(3)\Rightarrow (5)$ We have $\B+\C=\B+ \A(Y')$ and since $Y'$ is orthogonal to $X$ and $\A(Y')=c(\C,\B\cap\C)$, then $\B+\C=\B\oplus c(\C,\B\cap\C)$.

$(5)\Rightarrow (6)$ We take $\D=c(\C,\B\cap\C)$ and then $\D\subseteq\C$.

$(6)\Rightarrow (1)$ First $\D=c(\B+\C,\B)$ from the uniqueness of the orthogonal complement and $\D\subseteq\C$, so we get the result.

\end{proof}

\begin{prop}\label{p11xx}Let $\A\in\b_n$ and $\B,\C\in\b(\A)$. We shall say that $\B\simeq\C$ if the pair $(\B,\C)$ satisfies the star property. Then $\simeq$ is a reflexive and symmetric relation.
\end{prop}

\begin{proof}The reflexivity property of the relation $\simeq$ is trivial. Now suppose $\B\simeq\C$, then $\B+\C=(\B\cap\C)\oplus c(\B+\C,\B\cap\C)=(\B\cap\C)\oplus c(\B+\C,\B)\oplus c(\B,\B\cap\C)=(\B\cap\C)\oplus c(\C,\B\cap\C)\oplus c(\B,\B\cap\C)=\C\oplus c(\B,\B\cap\C)$, and since $\B+\C=\C\oplus c(\B+\C,\C)$. So $c(\B,\B\cap\C)=c(\B+\C,\C)$ and therefore $\C\simeq\B$.

\end{proof}

\begin{rem}\label{rr}The relation $\simeq$ is not transitive. Indeed Let $\A\in\b_n$ and $\B,\C\in\b(\A)$ such that the pair $(\B,\C)$ does not satisfy the star property, but the pairs $(\B+\C,\B)$ and $(\B+\C,\C)$ satisfy the star property.
\end{rem}

\begin{prop}\label{pppppp8}Let $\A\in\b_n$ and $\B,\C\in\b(\A)$. Suppose the pair $(\B,\C)$ satisfies the star property, then
\begin{enumerate}
\item the pair $(c(\A,\B),c(\A,\C))$ satisfies the star property.
\item the pair $(\B,c(\A,\C))$ satisfies the star property.
\end{enumerate}

\end{prop}
\begin{proof}(1) We have $c(\A,\B)+c(\A,\C)=c(\A,\B\cap\C)=c(\A,\B)\oplus c(\B,\B\cap\C)=c(\A,\B)\oplus c(\B+\C,\C)$ withy $c(\B+\C,\C)\subseteq c(\A,\C)$, then by assertion (4) in Proposition \ref{p7}, we get the result.

(2) We apply assertion (2) in Proposition \ref{p61} and get $c(\A,\C)=c(\A,\B)\oplus c(\B,\C)$ and then $\B+c(\A,\C)=\B+c(\A,\B)+c(\B,\C)=\A=\B\oplus c(\A,\B)$ and $c(\A,\B)\subseteq c(\A,\C)$, so the pair $(\B,c(\A,\C))$ satisfies the star property.
\end{proof}

The distributive laws are stated next.

\begin{prop}\label{p8dd}Let $\A\in\b_n$ and $\B,\C,\D\in\b(\A)$. Suppose the pairs $(\D,\B)$ and $(\D,\C)$ satisfy the star property, then
\begin{enumerate}
\item $\D+(\B\cap\C)=(\D+\B)\cap(\D+\C)$.
\item $\D\cap(\B+\C)=(\D\cap\B)+(\D\cap\C)$.
\end{enumerate}

\end{prop}
\begin{proof}
(1) We have $\D+(\B\cap\C)\subseteq(\D+\B)\cap(\D+\C)$. For the inverse inclusion, let $h\in \k((\D+\B)\cap(\D+\C))$ and since $\D+\B=\D\oplus\D_1$ and $\D+\C=\D\oplus\D_2$ with $\D_1\subseteq\B$ and $\D_2\subseteq\C$, then $h=f+g=f'+g'$ with $f,f'\in \k(\D)$, $g\in \k(\D_1)$ and $g'\in \k(\D_2)$. So $f-f'=g'-g$ with $f-f'\perp g'-g$ and therefore $g=g'\in\B\cap\C$ and $h=f+g\in\D+(\B\cap\C)$.

(2) We have $(\D\cap\B)+(\D\cap\C)\subseteq\D\cap(\B+\C)$. For the inverse inclusion, let $h\in \k(\D\cap(\B+\C))$, then $h=f+g\in\k(\D)$ with $f\in \k(\B)$ and $g\in \k(\C)$, and since $\B=\B\cap\D\oplus\D_1$ and $\C=\C\cap\D\oplus\D_2$ with $\D_1,\D_2\perp\D$, so $f=f_1+f_2$ and $g=g_1+g_2$ with $f_1\in\B\cap\D$, $g_1\in\C\cap\D$, $f_2\in\k(\D_1)$ and $g_2\in\k(\D_2)$. Therefore $h-f_1-g_1=g_2+f_2$ with $h-f_1-g_1\perp g_2+f_2$, then $f_2+g_2=0$ and by consequent $h=f_1+g_1\in (\D\cap\B)+(\D\cap\C)$.

\end{proof}

The star property in Proposition \ref{p8dd} is a necessary condition. Next we provide an example to illustrate that.

\begin{ex}\label{ex6}In a two dimensional Brownian setting we define $\D=\A(W^1)$, $\B=\A(W^1+W^2)$ and $\C=\A(W^1+2W^2)$. Then $c(\D+\B,\D)=\A(W^2)$ and $c(\D+\C,\D)=\A(W^2)$, so the pairs $(\D,\B)$ and $(\D,\C)$ do not satisfy the star property. Moreover $\D+(\B\cap\C)=\D$ and $(\D+\B)\cap(\D+\C)=\A(W^1,W^2)$ and $\A(W^1)\neq \A(W^1,W^2)$.

\end{ex}

\begin{cor}\label{cc3}Let $\A\in\b_n$ and $\B,\C,\D\in\b(\A)$. Suppose the pairs $(\D,\B)$ and $(\D,\C)$ satisfy the star property, then the pairs $(\D,\B\cap\C)$ and $(\D,\B+\C)$ satisfy the star property.

\end{cor}
\begin{cor}\label{pk1}Let $\A\in\b_n$ and $\B,\C\in\b(\A)$. Then the pair $(\B,\C)$ satisfy the star property iff $\C=(\C\cap\B)\op(\C\cap c(\A,\B))$.

\end{cor}
\begin{proof}The converse implication is a consequence of assertion (4) in Proposition \ref{p7}. For the direct one, we have that $\A=\B\oplus c(\A,\B)$ and $\C=\C\cap\A$, so by assertion (2) in Proposition \ref{pppppp8} and assertion (2) in Proposition \ref{p8dd}, we get the result.
\end{proof}


\subsection{Metric and correlation.}
As application of the random dimension $\dd$ of the plug-in vector space, we will introduce the notion of metric and correlation for pair of sets in $\b(\A)$ with $\A\in\b_n$ for some integer $n\geq 1$. In order to allege notation, we define the mapping $\mu(\B)=\te(\dd(\B))$ for $\B\in\b(\A)$.

\begin{prop}\label{p5}Let $\B,\C\in\b(\A)$ with $\A\in\b_n$ for some integer $n\geq 1$. Then
\begin{enumerate}
\item the relation defined by $\B\sim \C$ if $\dd(\B)=\dd(\C)$, is an equivalence relation.
\item the mapping $\varphi$ defined by $\varphi(\B,\C)={\tilde \e}\left(|\dd(\B)-\dd(\C)|\right)$, is a metric w.r.t. the equivalence relation $\sim$.
\item the mapping $\eta$ defined by $\eta(\B,\C)=\mu(\B+\C)-\mu(\B\cap\C)$, is a metric.
\item $\eta(\B,\C)=2\mu(\B+\C)-\mu(\B)-\mu(\C)$.
\item $\eta(\B,\C)=\mu(\B)+\mu(\C)-2\mu(\B\cap\C)$.
\item $\eta(\B,\C)=\varphi(\B+\C,\B)+\varphi(\B+\C,\C)$.
\item $\varphi(\B,\C)\leq\eta(\B,\C)$.
\item $\varphi(\B,\C)=\eta(\B,\C)$ if $\B\subseteq\C$.
\end{enumerate}

\end{prop}

\begin{proof}Assertions (1), (2) and (6) are easy to show. Assertions (4) and (5) are immediate consequences of Proposition \ref{ppp}.

For (3), we show first that $\eta$ is reflexive. $\eta(\B,\B)=0$ and if $\eta(\B,\C)=0$ for some $\B,\C\in\b(\A)$, then $\B+\C=\B\cap\C$ thanks to Proposition \ref{pe1} and therefore $\B=\C$. For the transitivity property, let $\B,\C,\D\in\b(\A)$ and define $q:=\eta(\B,\D)+\eta(\D,\C)-\eta(\B,\C)$, then by assertion (5) we get
$$
q=\mu(\B)+\mu(\D)-2\mu(\B\cap\D)+\mu(\D)+\mu(\C)-2\mu(\D\cap\C)-\mu(\B)-\mu(\C)+2\mu(\B\cap\C)
$$
$$
=2\mu(\D)-2\mu(\B\cap\D)-2\mu(\D\cap\C)+2\mu(\B\cap\C).
$$
Since $\mu(\B\cap\D)+\mu(\D\cap\C)=\mu((\B\cap\D)+(\D\cap\C))+\mu(\B\cap\D\cap\C)$ with $((\B\cap\D)+(\D\cap\C))\subseteq\D$ and $(\B\cap\D\cap\C)\subseteq\B\cap\C$, therefore $q\geq 0$.

Assertions (7) and (8) are immediate consequences of assertion (6).

\end{proof}

Now we introduce the correlation coefficient.

\begin{defi}\label{d6}
Let $\B,\C\in\b(\A)$ with $\A\in\b_n$ for some integer $n\geq 1$. We define the correlation coefficient of $\B$ and $\C$ by:
$$
\Corr(\B,\C)=\dfrac{\mu(\B\cap\C)}{\mu(\B+\C)},
$$
if $\mu(\B+\C)>0$ and $\Corr(\B,\C)=0$ if not.

\end{defi}

\begin{thm}\label{t6}
Let $\B,\C\in\b(\A)$ with $\A\in\b_n$ for some integer $n\geq 1$. Then we have the following:
\begin{enumerate}
\item $\Corr(\B,\C)\in [0,1]$.
\item $\Corr(\B,\C)=\Corr(\C,\B)$.
\item $\Corr(\B,\C)=0$ iff $\B\cap\C=\Linf_-$.
\item $\Corr(\B,\C)=1$ iff $\B=\C$.

\end{enumerate}
\end{thm}

\begin{proof}(1) Since $0\leq \dd(\B\cap\C)\leq\dd(\B+\C)$, then $\Corr(\B,\C)\in[0,1]$. Assertions (2) and (3) are trivial.

(4) The converse implication is trivial, let us show the direct one. We have that $\B+\C=\B\cap\C$ thanks to Proposition \ref{pe1}. So $\B=\C$.

\end{proof}

\begin{cor}\label{c7}Let suppose the vector process $(X,Y)\in\cc$ satisfy the property $(\gamma)$. Then $\varphi(\A(X),\A(Y))=|r(X)-r(Y)|$, $\eta(\A(X),\A(Y))=r(X)+r(Y)$ and $\rho(\A(X),\A(Y))=0$.

\end{cor}

\begin{cor}\label{c22}Let $\B,\C,\D\in\b(\A)$ with $\A\in\b_n$ for some integer $n\geq 1$ and suppose that $\D\subseteq\C\subseteq\B$. Then $\Corr(\B,\D)=\Corr(\B,\C)\Corr(\C,\D)$.

\end{cor}


\section{Application to incomplete markets in the Brownian setting.}

\subsection{Degree of completeness.}

Going back to the idea of measuring the degree of incompleteness of an incomplete financial market, we need first to assure that a complete market exists. Thanks to Theorem \ref{t5}, we should work in a setting where the filtration satisfies the martingale representation property and then without loss of generality we consider working in a Browian setting along this section. A variety of models for the price process are built in this framework, namely the well known Black-Scholes model and the family of stochastic volatility models. We consider an $\R^n$-valued Browian motion $B=(B^1,\ldots,B^n)$ with the associated filtration. In fact we follow the idea discussed briefly in the introduction and assume the following assumption: (i) all financial risks are modelled by a filtration $\bF:=(\F_t)_{t\in\T}$.
(ii) the filtration $\bF$ is generated by a Brownian motion.

The question is to determine how much gap the filtration $\bF^X$ is filling, where $X$ is the proposed model for the price process.

We will show first some properties in this setting, in particular that the property $(\gamma)$ and maximality property in $\b$ are equivalent.

\begin{thm}\label{t1}We have the following for $r\leq n$:
\begin{enumerate}
\item For any $\bq\in\p$, the set $\A^{\bq}:=\{X\in\Linf:\;\eq(X)\leq 0\}$ is maximal in $\b_n$.
\item Any process $X\in\cc_r$ satisfies the property $(\tgamma)$.
\item A process $X\in\cc_r$ satisfies the property $(\gamma)$ if and only if $\A(X)$ is maximal in $\b_r$.
\item For any set $\B\in\b_r$, there exits some maximal set $\A$ in $\b_r$ containing $\B$.
\item A process $X\in\cc_n$ satisfies the property $(\gamma)$ if and only if $\q^X$ is reduced to a singleton.

\end{enumerate}

\end{thm}

\begin{proof}(1) We know that the density process $Z$ of $\bq$, is the solution of the stochastic differential equation $dZ/Z=\lam.dW$ for some $\lam\in\m_{1,n}(\cp)$. So the process $B^{\bq}=B-\int_0^.\,\lam_s\,ds$ is a $\bq$-Brownian motion, which means that $\A^{\bq}=\A(B^{\bq})$ and then $\A^{\bq}\in\b_n$ and therefore $\A^{\bq}$ is maximal in $\b_n$.

(2) Let $X\in\cc_r$ and $\bq\in\q^X$ with $\bq\sim\P$. Then $\A(X)\subseteq\A^{\bq}=\A(B^{\bq})$ and the process $B^{\bq}$ satisfies the property $(\gamma)$.

(3) It is an immediate consequence of Proposition \ref{p22}.

(4) Let $\B\in\b_r$, so $\B=\A(Y)$ for some $Y\in\cc_r$ which satisfies the property $(\tgamma)$ thanks to assertion (3). Thanks to Proposition \ref{pt1}, there exists some $X\in\cc_r$ satisfying the property $(\gamma)$ such that $\A(Y)\subseteq\A(X)$ and $\A(X)$ is maximal in $\b_r$ thanks to assertion (3).

(5) Thanks to assertion (3), the process $X\in\cc_n$ satisfies the property $(\gamma)$ if and only if $\A(X)$ is maximal in $\b_n$, which is equivalent to $\q^X$ being reduced to a singleton.

\end{proof}

Now we define the two proportional average degrees of completeness and incompleteness of an incomplete financial market.

\begin{defi}\label{dd1}Let $X\in\cc$. We define respectively the proportional average degrees of completeness and incompleteness of the financial market generated by $X$, by
$$
\delta_c(X)=\dfrac{\varphi(\A(X),\Linf_-)}{\varphi(\A(B),\Linf_-)}={\tilde \e}\left(\dfrac{\dd(X)}{n}\right),
$$
and
$$
\delta_i(X)=\dfrac{\varphi(\A(X),\A(B))}{\varphi(\Linf_-,\A(B))}=1-\delta_c(X),
$$
where the metric $\varphi$ is defined in Proposition \ref{p5}.
\end{defi}

Some properties of $\delta_c$ are stated below.

\begin{prop}\label{p3}Let $(X,Y)\in\cc$. Then
\begin{enumerate}
\item $\delta_c(X)\leq \delta_c(X,Y)$.
\item $\delta_c(X)= \delta_c(X,Y)$ if and only if $\A(Y)\subseteq\A(X)$.
\item $\delta_c(X)=1$ if and only if $\q^X$ is reduced to a singleton.

\end{enumerate}

\end{prop}

\begin{proof}The assertions (1) and (2) are immediate consequences of the definition. (3) $\delta_c(X)=1$ $\Leftrightarrow$ $\dd(X)=n$ since $\dd(X)\leq n$ $\Leftrightarrow$ $\tp(\phi_n(X))=1$ $\Leftrightarrow$ $\q^X$ is reduced to a singleton.

\end{proof}

\begin{rem}\label{r6}Let $(X,Y)\in\cc$ such that $\A(X)\subseteq\A(Y)$. Then $\delta_c(X)=\Corr(\A(X),\A(Y))\delta_c(Y)$. In particular $\delta_c(X)=\Corr(\A(X),\A(B^{\bq}))$ for any $\bq\in\Ml^e(X)$ and $B^{\bq}$ is the associated $\bq$-Brownian motion.

\end{rem}

\subsection{Hedging process.}
Finally we look at the hedging process of a contingent claim. First we state a martingale decomposition theorem in this context.

\begin{thm}\label{t22'}Let $X\in\cc$ and $\bq\in\Ml^e(X)$. Then
\begin{enumerate}
\item there exists $X'\in\cc$ which is orthogonal to $X$ such that $\{\bq\}=\Ml(X,X')$ and $\dd(X)+\dd(X')=n$.
\item for any local $\bq$-martingale $L$, there exists two vector predictable processes $\al$ and $\al'$ such that $L=\al\bullet X+\al'\bullet X'$.
\end{enumerate}
\end{thm}

\begin{proof}It is an immediate consequence of Proposition \ref{p7}.

\end{proof}

\begin{thm}\label{t22}Let $X\in\cc$ and $\bq\in\Ml^e(X)$. Then for all $h\in\L^1(\bq)$, there exists two vector predictable processes $\al$ and $\al'$ such that $h=\eq(h)+L_T$ and $L=\al\bullet X+\al'\bullet X'$, where $X'$ is given in Theorem \ref{t22'}.

\end{thm}

\begin{proof}We define the process $L=\eq(h)$ and apply Theorem \ref{t22'} to deduce the result.

\end{proof}


\end{document}